\documentclass[a4paper,12pt,final]{article}
 \usepackage[T1]{fontenc}
  \usepackage[utf8]{inputenc}
\usepackage[sfdefault]{biolinum}  
\usepackage{amsfonts,amsmath,amsthm,amssymb,xcolor,colortbl}
\usepackage{enumitem}
\usepackage{upgreek}
\usepackage{geometry}
\usepackage{setspace}
\usepackage{soul}
\usepackage{setspace}
\usepackage{natbib}
\usepackage{clipboard}	
\newclipboard{PaperClipboard}	
\usepackage{xr-hyper} 
\usepackage[colorlinks,linkcolor=links,citecolor=cites,urlcolor=MyDarkBlue]{hyperref}
\usepackage{cleveref}
\usepackage{bbm}
\usepackage[justification=centering]{caption}
\usepackage[justification=centering]{subfig}
\usepackage[section]{placeins}
\usepackage{parskip}
\usepackage{nicefrac}
\usepackage{appendix}
\usepackage[multiple]{footmisc}
\definecolor{MyDarkBlue}{rgb}{0,0.08,0.45}
\definecolor{cites}{HTML}{324b13}
\definecolor{links}{HTML}{1a663b}
\definecolor{MyLightMagenta}{cmyk}{0.1,0.8,0,0.1}
\definecolor{sblue}{HTML}{0049A9}
\definecolor{scyan}{HTML}{CBEAFC}
\definecolor{sred}{HTML}{B5595C}
\definecolor{sgreen}{HTML}{609B57}
\definecolor{spink}{HTML}{FFB0FF}

\usepackage{tikz}
\usepackage{pgfplots}
\usepgfplotslibrary{fillbetween}
\usepgfplotslibrary{patchplots}
\usetikzlibrary{positioning,chains,fit,shapes,calc,arrows,trees,decorations.pathmorphing,decorations.markings,hobby,calc,snakes,shapes.misc,patterns,backgrounds,quotes}
\usetikzlibrary{arrows}
\usetikzlibrary{trees,calc,snakes,fit,shapes,positioning}
\usetikzlibrary{decorations.pathmorphing}
\usetikzlibrary{decorations.markings}
\usetikzlibrary{decorations.shapes}
\tikzset{cross/.style={cross out, draw=black, minimum size=2*(#1-\pgflinewidth), inner sep=0pt, outer sep=0pt},
cross/.default={1pt}}
\newtheoremstyle{ex}
{3pt}
{3pt}
{}
{}
{\bfseries}
{.}
{.5em}
{}%
\newtheoremstyle{rex}
{}
{}
{}
{}
{\bfseries}
{.}
{.5em}
{}%
\theoremstyle{ex}\newtheorem{example}{Example}

 \tikzset{ every node/.style={inner sep=0pt,minimum size=1mm},
  nsnode/.style={draw,circle,black},
    nsnode2/.style={draw,circle,blue},
  nnnode/.style={draw,circle,black,fill=black},
  asnode/.style={draw,circle,blue},
  bsnode/.style={draw,circle,green,fill=green},
  csnode/.style={draw,circle,red,fill=red, minimum size=2mm},
  every fit/.style={inner sep=-1.5pt,text width=1cm}  }
  \tikzset{nero/.style={decorate,draw=black}}
\tikzset{bianco/.style={decorate,draw=bg}}

\theoremstyle{theorem}\newtheorem{theorem}{Theorem}
\theoremstyle{theorem} 
\theoremstyle{theorem}\newtheorem{proposition}{Proposition}
 
\theoremstyle{theorem}\newtheorem{observation}{Observation} 

\theoremstyle{theorem}\newtheorem{definition}{Definition} 
\theoremstyle{theorem}\newtheorem{corollary}{Corollary}

\newtheorem{remark}{Remark}
\theoremstyle{theorem}\newtheorem{claim}{Claim}

\newcommand{\bs}[1]{\mathrm{#1}}
\newcommand{\T}[1]{\ensuremath{#1^{T}}}

\newcommand{\type}{\ensuremath{\theta}}
\newcommand{\typeb}{\ensuremath{\type^\prime}}
\newcommand{\typec}{\ensuremath{\tilde{\type}}}

\newcommand{\types}{\ensuremath{\Theta}}
\newcommand{\Types}{\ensuremath{\Theta}}

\newcommand{\event}{\ensuremath{X}}

\newcommand{\policy}{\ensuremath{\Pi}}
\newcommand{\signal}{\ensuremath{s}}
\newcommand{\signals}{\ensuremath{S}}

\newcommand{\reals}{\ensuremath{\mathbb{R}}}

\newcommand{\alloc}{\ensuremath{w}}
\newcommand{\allocv}{\ensuremath{\mathrm{w}}}
\newcommand{\allocvi}{\ensuremath{\mathrm{w}_i}}
\newcommand{\allocvF}{\ensuremath{\mathrm{w}^{FD}}}
\newcommand{\allocvN}{\ensuremath{\mathrm{w}^{ND}}}
\newcommand{\gammav}{\ensuremath{\upgamma}}
\newcommand{\gammavi}{\ensuremath{\gammav_i}}
\newcommand{\gammavj}{\ensuremath{\gammav_j}}

\newcommand{\repset}{\ensuremath{\mathrm{W}}}

\newcommand{\expect}{\ensuremath{\mathbb{E}}}
\newcommand{\aaction}{\ensuremath{a}}
\newcommand{\actions}{\ensuremath{A}}
\newcommand{\informationpolicy}{information structure}
\newcommand{\informationpolicies}{information structures}

\newcommand{\Informationpolicies}{Information structures}

\newcommand{\cvector}{\ensuremath{\mathrm{c}}}
\newcommand{\eventcorr}{\ensuremath{\upbeta}}
\newcommand{\transmati}{\ensuremath{\mathrm{\transmat_{i\cdot}}}}
\newcommand{\cpmathat}{\mathrm{\hat{\cpmat}}}
\newcommand{\dprior}{\ensuremath{\mathrm{D_0}}}
\newcommand{\zetav}{\ensuremath{\mathrm{\zeta}}}
\newcommand{\informationprofile}{welfare profile}
\newcommand{\informationprofiles}{welfare profiles}

\newcommand{\bipooling}{noisy-priority policy}

\newcommand{\typel}{\ensuremath{\theta_L}}
\newcommand{\typeh}{\ensuremath{\theta_H}}

\newcommand{\transmat}{\ensuremath{\mathrm{P}}}
\newcommand{\cpmat}{\ensuremath{\mathrm{C}}}

\newcommand{\targettype}{\ensuremath{\type_i}}

\newcommand{\linear}{\ensuremath{\rho}}
\newcommand{\linearv}{\ensuremath{\uprho}}
\newcommand{\linearvi}{\ensuremath{\linearv_i}}
\newcommand{\linearvj}{\ensuremath{\linearv_j}}

\newcommand{\Reputation}{\ensuremath{\hat{w}}}
\newcommand{\Reputationv}{\ensuremath{\mathrm{\hat{w}}}}

\newcommand{\reputation}{\ensuremath{w}}

\newcommand{\allocbv}{\ensuremath{\mathrm{\alloc}^\prime}}
\newcommand{\welfare}{\ensuremath{w}}
\newcommand{\welfarev}{\ensuremath{\mathrm{w}}}

\newcommand{\excdf}{\ensuremath{\pi}}
\newcommand{\payoff}{\ensuremath{w}}
\newcommand{\payoffv}{\ensuremath{\mathrm{w}}}

\newcommand{\vectors}{\ensuremath{x}}
\newcommand{\vectorsv}{\ensuremath{\mathrm{x}}}
\newcommand{\vectorsvone}{\ensuremath{\mathrm{x_1}}}
\newcommand{\vectorsvm}{\ensuremath{\mathrm{x_m}}}
\newcommand{\vectorsvM}{\ensuremath{\mathrm{x_M}}}
\newcommand{\vectorsvT}{\ensuremath{\T{\vectorsv}}}
\newcommand{\repsetbp}{\ensuremath{\repset_{\mathrm{BP}}}}

\newcommand{\plotcolor}{blue}

\newcommand{\Posteriors}{\ensuremath{\Delta(\types)}}

\newcommand{\prior}{\ensuremath{\mu_0}}
\newcommand{\priorv}{\ensuremath{\upmu_{\mathrm{0}}}}

\newcommand{\ones}{\ensuremath{\mathrm{e}}}

\newcommand{\direction}{\ensuremath{\lambda}}
\newcommand{\directionv}{\ensuremath{\uplambda}}

\newcommand{\bsplit}{\ensuremath{\tau}}
\newcommand{\belief}{\ensuremath{\mu}}
\newcommand{\beliefv}{\ensuremath{\upmu}}
\newcommand{\beliefvm}{\ensuremath{\beliefv_m}}

\newcommand{\beliefvmT}{\ensuremath{\T{\beliefv}_m}}

\newcommand{\beliefc}{\ensuremath{\tilde{\belief}}}

\newcommand{\selection}{\ensuremath{\alpha}}

\newcommand{\data}{\ensuremath{d}}
\newcommand{\Data}{\ensuremath{D}}
\newcommand{\Datab}{\ensuremath{\Data^\prime}}

\newcommand{\datab}{\ensuremath{\data^\prime}}
\newcommand{\posteriordata}{\ensuremath{\eta}}
\newcommand{\precision}{\ensuremath{\sigma}}

\newcommand{\weta}{\ensuremath{w_{\dagger}}}
\newcommand{\wetahat}{\ensuremath{\hat{w}_{\dagger}}}
\newcommand{\wetahatv}{\ensuremath{\mathrm{\hat{w}_{\dagger}}}}

\newcommand{\maxwelfare}{\ensuremath{\welfare^*}}
\newcommand{\minwelfare}{\ensuremath{\welfare_*}}

\newcommand{\vex}{\ensuremath{\mathrm{vex}}}
\newcommand{\cav}{\ensuremath{\mathrm{cav}}}

\newcommand{\Left}{\ensuremath{A}}
\newcommand{\Right}{\ensuremath{B}}
\newcommand{\typea}{\ensuremath{\type_\Left}}
\newcommand{\typebb}{\ensuremath{\type_\Right}}

\newcommand{\agent}{individual}
\newcommand{\agents}{individuals}
\newcommand{\bprofile}{Bayes welfare profile}
\newcommand{\bprofiles}{Bayes welfare profiles}
\newcommand{\bset}{Bayes welfare set}
\newcommand{\bsets}{Bayes welfare sets}
\newcommand{\joint}{\ensuremath{\nu}}
\newcommand{\jointprior}{\ensuremath{\joint_0}}

\newcommand{\jointpriorb}{\ensuremath{\jointprior^\prime}}

\newcommand{\beliefdata}{\ensuremath{\eta}}
\newcommand{\priordata}{\ensuremath{\beliefdata_0}}

\newcommand{\closing}{\ensuremath{\lozenge}}

\newcommand{\characteristic}{type}
\newcommand{\characteristics}{types}
\newcommand{\unconditional}{\ensuremath{\langle\policy\rangle}}
\newcommand{\conditional}{\ensuremath{\langle\policy|\type\rangle}}

\newcommand{\canonical}{\ensuremath{\mathrm{t}}}
\newcommand{\canonicali}{\ensuremath{\canonical_i}}

\newcommand{\bplausible}{\ensuremath{\Delta_{\prior}\left(\Posteriors\right)}}
\newcommand{\Pareto}{\ensuremath{\repset_{\mathrm{P}}}}

\newcommand{\vcorr}{\ensuremath{\hat{V}}}
\newcommand{\weights}{\ensuremath{\tau}}
\newcommand{\vrep}{\ensuremath{\hat{\mathrm{v}}}}

\usepackage{mparhack}

  \usepackage{chngcntr}
\usepackage{apptools}
\AtAppendix{\counterwithin{equation}{section}}
\AtAppendix{\counterwithin{figure}{section}}
\AtAppendix{\counterwithin{table}{section}}
\AtAppendix{\counterwithin{example}{section}}
\AtAppendix{\counterwithin{remark}{section}}
\AtAppendix{\counterwithin{proposition}{section}}
\AtAppendix{\counterwithin{theorem}{section}}
\AtAppendix{\counterwithin{corollary}{section}}
\usepackage{etoolbox}
\makeatletter
\patchcmd{\hyper@makecurrent}{%
    \ifx\Hy@param\Hy@chapterstring
        \let\Hy@param\Hy@chapapp
    \fi
}{%
    \iftoggle{inappendix}{
        \@checkappendixparam{chapter}%
        \@checkappendixparam{section}%
        \@checkappendixparam{subsection}%
        \@checkappendixparam{subsubsection}%
        \@checkappendixparam{paragraph}%
        \@checkappendixparam{subparagraph}%
    }{}%
}{}{\errmessage{failed to patch}}

\newcommand*{\@checkappendixparam}[1]{%
    \def\@checkappendixparamtmp{#1}%
    \ifx\Hy@param\@checkappendixparamtmp
        \let\Hy@param\Hy@appendixstring
    \fi
}
\makeatletter

\newtoggle{inappendix}
\togglefalse{inappendix}

\apptocmd{\appendix}{\toggletrue{inappendix}}{}{\errmessage{failed to patch}}
\apptocmd{\subappendices}{\toggletrue{inappendix}}{}{\errmessage{failed to patch}}

\title{Persuasion and Welfare\thanks{
We thank the Editor, Emir Kamenica, and three anonymous referees for feedback that has greatly improved this paper. For valuable suggestions and comments, we would like to thank Ricardo Alonso, Odilon C\^amara, Navin Kartik, Elliot Lipnowski, Antonio Penta,  Jean Tirole, and Kai Hao Yang, as well as seminar participants at Toulouse School of Economics, Columbia, Bonn Winter Theory Workshop 2021, Warwick Theory Workshop 2022, Stony Brook 2022, ESSET 2022, EEA-ESEM 2022, and Virtual Seminars in Economic Theory. We thank Shunsuke Matsuno for excellent research assistance. 
Smolin acknowledges funding from the French
National Research Agency (ANR) under the Investments for the Future (Investissements d'Avenir) program (grant
ANR-17-EURE-0010).}}
\author{Laura Doval\thanks{Columbia Business School and CEPR. E-mail: \href{mailto:laura.doval@columbia.edu}{\texttt{laura.doval@columbia.edu}.}} \and Alex Smolin\thanks{Toulouse School of Economics and CEPR. E-mail: \href{mailto:alexey.v.smolin@gmail.com}{\texttt{alexey.v.smolin@gmail.com}.}}}
\begin{document}
\pagenumbering{gobble}
\maketitle
\begin{abstract}
Information policies such as scores, ratings, and recommendations are increasingly shaping society's choices in high-stakes domains. We provide a framework to study the welfare implications of information policies on a population of heterogeneous \agents. We define and characterize the \emph{\bset}, consisting of the population's utility profiles that are feasible under some information policy. The Pareto frontier of this set can be recovered by a series of \emph{standard Bayesian persuasion} problems, in which a utilitarian planner takes the role of the information designer. We provide necessary and sufficient conditions under which an information policy exists that Pareto dominates the no-information policy. We illustrate our results with applications to data leakage, price discrimination, and credit ratings.



\end{abstract}

\newpage
\clearpage
\pagenumbering{arabic}
\section{Introduction}\label{sec:intro}


 Information has increasingly become a tool for shaping society's choices in high-stakes domains. Consider, for instance, the role of algorithms in making recommendations for bail \citep{angwin2016machine}, recruiting \citep{raghavan2020mitigating,li2020hiring}, health \citep{obermeyer2019dissecting}, education \citep{kuvcak2018machine}, and lending \citep{jagtiani2019roles}, among others. The role of information as a policy instrument is not confined to the \emph{big data} economy. Indeed, information policies in the form of scores and ratings have been in place long before algorithmic recommendations to determine school placement, promotions, and who receives credit.  As society becomes reliant on information to guide decisions in policy-relevant domains, understanding the welfare implications of information-based policies becomes a first-order concern.
 
In this paper, we provide a framework to study the welfare impact of information policies in a population of heterogeneous agents. Formally, we study the following model. There is a unit mass population with types in a finite set, distributed according to a given prior distribution. We model an information policy as an \informationpolicy, which associates to each type a distribution over signals and hence, via Bayes' rule, a distribution over posterior beliefs \citep{kage11}. Our primitive is a welfare function that represents for each (posterior) type distribution the welfare of \agents\ of a given type. Given the welfare function, each \informationpolicy\ induces a \emph{\bprofile}, which describes for each type in the population their expected payoff under the \informationpolicy. We define the \bset\ to be the set of all such profiles.

We characterize the \bset\ and study its properties. The \bset\ allows us to reduce society's choice of an \informationpolicy\ to the choice of a \bprofile. Instead of imposing properties that the \informationpolicy\ must satisfy, society's preferences over the welfare distribution in the population determine the properties of the chosen \informationpolicy. Our perspective thus complements that of the algorithmic fairness literature which remains agnostic about the population's payoffs, focusing instead on statistical properties of \informationpolicies, such as accuracy, parity, or fairness. It also complements that of the literature in Bayesian persuasion \citep{rayo2010optimal,kage11}, which characterizes the (maximum) average welfare consistent with some information structure, but not necessarily the \bprofiles\ that give rise to the average welfare.


\autoref{theorem:concavification} characterizes the \bset\ via the convex hull of a vector-valued function. In doing so, we extend the geometric characterizations of \citet{auma95} and \citet{kage11} of the feasible set of ex ante payoffs to the characterization of the \bset. Whereas a \bprofile\ depends on the distribution over posteriors conditional on each type, we show it can be alternatively expressed as the unconditional expectation over posteriors of a \emph{truth-adjusted} payoff function, where the adjustment is proportional to the posterior likelihood ratio of each type. Evaluated at a given type, the truth-adjusted welfare function allows us to characterize the welfare \agents\ of a given type may obtain under some information structure. In turn, interpreting the truth-adjusted payoff function as a vector-valued function allows us to capture the across-type restrictions imposed by Bayes' rule and precisely characterize the \bset.

\autoref{theorem:hyperplanes} characterizes the Pareto frontier of the \bset. Points in the Pareto frontier are natural candidates for being the outcome of efficient bargaining over \informationpolicies\ or a social planner's choice. \autoref{theorem:hyperplanes} shows the points in the Pareto frontier of the \bset\ can be recovered by a series of \emph{standard} Bayesian persuasion problems, in which a utilitarian planner takes the role of an information designer. We use \autoref{theorem:hyperplanes} throughout the paper to characterize optimal \informationpolicies\ in specific applications. Leveraging \autoref{theorem:hyperplanes}, \autoref{proposition:pareto-improvement} provides a necessary and sufficient condition under which an \informationpolicy\ exists that Pareto dominates providing no information.


\autoref{theorem:linear-concavification} enriches our characterization in the case in which the welfare function is equal to the expectation of a one-dimensional random variable, the support of which we call the \emph{reputation vector}. This special case constitutes a natural benchmark and is commonly used in the literature on career concerns (\citealp{holm99}), social image (\citealp{beti06}, \citealp{tirl20}), and policy prediction problems \citep{mullainathan2018algorithmic}. \autoref{theorem:linear-concavification} shows  a welfare profile belongs to the \bset\ if and only if it can be represented as the product between the reputation vector and a \emph{completely positive matrix}  that satisfies a version of Bayes plausibility.\footnote{A matrix $\cpmat\in\mathbb{R}^{N\times N}$ is completely positive if non-negative vectors $\cvector_1,\dots,\cvector_K\in\mathbb{R}_+^N$ exist such that $\cpmat=\sum_{i=1}^K\cvector_i\cvector_i^T$ \citep{berm88}.} In addition, we show that the Bayesian persuasion problems that characterize the \emph{relative} boundary of the \bset\ correspond to instances of the problem in \cite{rayo2010optimal}. It follows that the \informationpolicies\ that induce welfare profiles on the relative boundary of the \bset\ can be characterized using the graph-theoretic approach in \cite{rayo2010optimal}. We leverage their approach in \autoref{proposition:bi-pooling}, where we show that the \informationpolicies\ that maximize the welfare of a given type in the population correspond to a noisy version of the priority mechanisms studied in the matching literature (e.g., \citealp{celebi2022adaptive}).

Finally, we note that by interpreting our welfare function as an \agent's type-dependent payoff function, the \bset\ is also the object of interest in more standard information design applications. For instance, the types may represent the private information of an informed principal who can commit to an \informationpolicy\ only \emph{after} observing her type, as in \cite{perez2014interim} and \cite{koessler2021information}. Similarly, in the study of mechanism design with limited commitment, \cite{doval2022mechanism} describe the principal's mechanism as an \informationpolicy\ that must satisfy an informed agent's incentive constraints. Similar constraints appear in the studies of information design without commitment, as in \cite{frechette2019rules}, \cite{lipnowski2020cheap}, and \cite{salamanca2021value}, in the analysis of tests subject to participation constraints in \cite{rosar2017test}, and in the analysis of moral hazard in  \cite{saeedi2020optimal}. Thus, the \bset\ can be viewed as a unifying concept that underlies the incentive constraints the equilibrium \informationpolicy\ must satisfy. As we show in our first working paper version, \cite{doval2021information}, our tools also open the door to the study of new problems in this literature.

\paragraph{Related Literature:} Our work contributes to the literature on information design reviewed in the introduction. Starting from the work of \cite{kage11} and \cite{rayo2010optimal}, a series of papers investigate the limits imposed by common knowledge of Bayesian rationality \citep{auma87}. Whereas the Bayesian persuasion literature studies the (maximum) \emph{average} welfare that can be achieved under some \informationpolicy, we characterize instead the welfare profiles that are consistent with some information structure. 



Whereas ours is the first characterization of the \bset, a small literature studies certain \bprofiles\ within applications. Assuming the information designer is an informed principal, \cite{perez2014interim} refers to the payoff profile induced by an \informationpolicy\ as an \emph{interim} payoff and studies the informed principal's preferred \bprofile. Recently, \cite{galperti2022value} study the \bprofile\ that gives rise to the sender's maximum average payoff and relate it to the Lagrange multiplier in the Bayes plausibility constraint. Restricting attention to the case in which the welfare function is linear in beliefs, \cite{saeedi2020optimal} characterize a subset of the \bset\ that satisfies certain incentive compatibility constraints, thus obtaining a different characterization. Finally, as the analysis below makes clear, a \bprofile\ depends on the distribution over posteriors induced by an information structure \emph{conditional} on each type. Whereas \cite{levy2021feasible} and \cite{arieli2022persuasion} characterize the set of conditional distributions over posteriors consistent with the prior, we follow a complementary approach that allows us to carry only the (unconditional) distribution over posteriors induced by an \informationpolicy\ (see \autoref{claim:accounting}).

We also contribute to the economics literature that studies algorithmic fairness. \cite{mullainathan2018algorithmic}, \cite{kleinberg2018algorithmic}, and \cite{rambachan2020economic} argue for letting the social planner's objective determine the properties of algorithms. Our analysis is also related to \cite{liang2022algorithmic}. Starting from a fixed joint distribution over groups, covariates, and states, \cite{liang2022algorithmic} model an algorithm as taking actions directly as a function of covariates and study the group error profiles in the fairness-accuracy frontier as the algorithm varies. Instead, we model an algorithm as an information structure that sends non-binding action recommendations to an unmodeled receiver, whose actions determine the population's welfare and characterize the set of all \bprofiles\ as we vary the algorithm.

By considering the welfare redistribution effects of information, our work joins the mechanism design literature that studies the role of markets in redistributing welfare (see, e.g., \citealp{dworczak2021redistribution}, \citealp{akbarpour2020redistributive}, and \citealp{akbarpour2023economic}). We complement this work, which typically assumes the planner can utilize transfers to achieve its objectives, by considering the role of information, which can be a powerful tool when the planner does not have access to transfers.\label{page-r2-redistribution}


Finally, our work contributes indirectly to the literature on higher-order beliefs. Indeed, when the welfare function is linear in beliefs as in \autoref{sec:linear}, the welfare profile can be seen as a profile of second-order expectations. Starting with \cite{samet1998iterated}, a body of work uses Markov matrices to represent such higher-order beliefs and expectations of higher-order beliefs for a \emph{given} information structure (see, e.g., \citealp{cripps2008common,golub2017higher}). Instead, our result in \autoref{theorem:linear-concavification} identifies the set of matrices that correspond to \emph{some} information structure.

In lieu of an organizational paragraph, we summarize below the notation used throughout the paper:
\paragraph{Notation:} For ease of presentation, we sometimes find it convenient to denote a function from a set \Types\ to $\mathbb{R}$ as a vector in $\mathbb{R}^{N}$, where $N$ is the cardinality of \Types. In this case, we reserve the italic notation \vectors\ for the function $\vectors:\Types\mapsto\mathbb{R}$ and the upright notation \vectorsv\ for the vector in $\mathbb{R}^N$. Any vector $\vectorsv\in\mathbb{R}^N$ is taken to be a column vector; we denote its $i^{th}$ component by $\vectorsv_i$ or $\vectors(\type_i)$ interchangeably.
If $\vectorsv\in\mathbb{R}^N$ is a column vector, $\vectorsv^{T}$ denotes its transpose. If $\vectorsv,\mathrm{y}$ are two vectors, $\vectorsv*\mathrm{y}$ denotes their Hadamard (element-wise) product  and $\nicefrac{\vectorsv}{\mathrm{y}}$ denotes their Hadamard division.  We denote by $\ones\in\reals^N$ the vector with $\ones_1=\dots=\ones_N=1$. When we want to emphasize that $x$ is a random variable, we write it as $\tilde x.$

\section{Model}\label{sec:model}
A unit mass population has \characteristics\ in a finite set, $\types\equiv\{\type_1,\dots,\type_N\}$. Letting $\Posteriors$ denote the set of probability distributions over \types, we denote by $\prior\in\Posteriors$ the frequency of \characteristics\ in the population. We assume that \prior\ has full support. We denote by $\Delta(\Posteriors)$ the set of distributions over posteriors, and by $\bplausible$ the set of distributions over posteriors with mean equal to the prior \prior.

\paragraph{Welfare function:} An \agent's welfare depends on her type \type\ and an (unmodeled) outside observer's belief about her type. We represent this by a welfare function $\welfare:\Posteriors\times\types\mapsto\reals$ that represents for each belief \belief\ and each type \type, the welfare of \agents\ of type \type\ under belief \belief, $\welfare(\belief,\type)$. We assume throughout that \welfare\ is bounded.

A welfare profile is a vector $\allocv\in\reals^N$, where $\allocvi$ describes the welfare level of individuals with type $\type_i$. Any \informationprofile\ \allocv\ induces an ex ante welfare of $\priorv^T\allocv=\sum_{i=1}^N\prior(\type_i)\welfarev_i$. When we average a profile \welfarev\ using weights other than the prior \prior, we refer to \emph{average} welfare instead. We are interested in characterizing those welfare profiles that are induced by some \informationpolicy.

\paragraph{\Informationpolicies:} An \informationpolicy\ $\policy=(\excdf,\signals)$ consists of a countable set of labels \signals, and a mapping \excdf, which associates to each type \type\ a distribution over signals $\excdf(\cdot|\type)\in\Delta(\signals)$. Given an \informationpolicy\ $\policy$ and a signal realization $\signal\in\signals$, the corresponding posterior belief $\belief_\signal\in\Posteriors$ is obtained by Bayes' rule whenever possible, and is given by
\begin{align*}
\belief_\signal(\type)=\frac{\prior(\type)\excdf(\signal|\type)}{\sum_{\typeb\in\types}\prior(\typeb)\excdf(\signal|\typeb)}.
\end{align*}
 Thus, an \informationpolicy\ can be seen as inducing a distribution over posterior beliefs $\{\belief_s:s\in\signals\}$. In what follows, two such distributions are of interest: the  distribution over posterior beliefs \emph{conditional} on an \agent's type--as induced by $\excdf(\cdot|\type)$--and the \emph{unconditional} distribution over posterior beliefs--as induced by the prior \prior\ and the signal distribution. When taking expectations using these distributions, we use the notations \conditional\ and \unconditional\ to denote the conditional and unconditional distributions over posterior beliefs, respectively.
\paragraph{\bprofiles:}The welfare function \welfare\ together with an \informationpolicy, \policy, defines a \informationprofile, $\welfare_\policy:\Types\mapsto\reals$, as 
\begin{align}\label{eq:interim-profile}
\welfare_{\policy}(\type)\equiv\expect_{\conditional}[\welfare(\beliefc, \type)]=\sum_{\signal\in\signals}\excdf(\signal|\type)\welfare(\belief_\signal,\type).
\end{align}    
That is, for each type \type, $\welfare_\policy(\type)$ describes the (expected) welfare of type-\type\ \agents\  under \informationpolicy\ \policy. Note that in computing the welfare of type-\type\ \agents, their type \type\ enters twice: directly through the welfare function, $\welfare(\cdot,\type)$, and indirectly through the signal distribution, $\excdf(\cdot|\type)$.

We now  present our two main objects of study:
\begin{definition}[\bprofile]\label{definition:interim-profile}
A welfare profile $\welfarev\in\reals^N$ is a \emph{\bprofile} if an \informationpolicy, \policy, exists such that for all types $\type_i$, $\welfarev_i=\welfare_\policy(\type_i)$.
\end{definition}

\begin{definition}[\bset]\label{definition:interim-set}
The \bset\ is the set of all Bayes welfare profiles; that is, 
\begin{align}\label{eq:interim-set}
\repset\equiv\left\{\allocv\in\mathbb{R}^N: \exists\,\policy \mathrm{\ s.t.\ }\allocv_i=\welfare_\policy(\type_i)\ \forall i\in\{1,\dots,N\}\right\}.
\end{align}
\end{definition}
The \bset\ \repset\ represents the utility possibility set in an economy where the allocations are given by information structures. As such, it describes the welfare effects that different \informationpolicies\ have for \agents\  with different \characteristics\ in applications such as grading schemes in the case of schooling (\citealp{ostrovsky2010information}), disclosure about job performance (\citealp{mukherjee2008sustaining}), affirmative action in the case of college admissions or the job market, rating systems in the case of platforms \citep{saeedi2020optimal}, and market segmentations \citep{bergemann2015limits}. 

Throughout, we illustrate our results using the following examples:
\color{black}
\begin{example}[Data leakage]\label{example:privacy} 
Consumers concerned about how a third party may use their data wish to maximize the third party's uncertainty about their types.\footnote{Because maximizing uncertainty can sometimes be accomplished by providing some information, this notion of data privacy differs from an unambiguous preference for no information disclosure.} We formalize this as follows. There are two types of consumer, $\Types=\{\typea,\typebb\}$. Letting $\prior$ denote the frequency of consumers of type \typebb\ (i.e., $\prior\equiv\prior(\typebb)$), we assume that $\prior=0.3$.  When the third party believes the consumer's type is \typebb\ with probability $\belief\in[0,1]$, a consumer experiences a welfare loss of $\welfare(\belief,\type)=-(\mu-\nicefrac{1}{2})^2$. The consumer's welfare loss is minimal when the third party is maximally confused about the consumer's type (i.e., $\belief=\nicefrac{1}{2}$). Instead, the consumer's welfare loss is maximal when the third party has precise information about the consumer's type (i.e., $\belief\in\{0,1\}$). 
 \hfill$\closing$
\end{example}
\begin{example}[Price discrimination]\label{example:platforms}
 An online marketplace makes algorithmic recommendations to a seller about what price the seller should set for consumers. There are two consumer types, \typeh\ and  \typel. Assume $\prior\equiv\prior(\typel)=0.6$. Consumers can have one of three values for the seller's good: low ($1$), medium ($2$), and high ($3$). A consumer's type indexes their distribution over values, with \typeh-consumers being more likely to have high valuations, and \typel-consumers being more likely to have low valuations. In particular, we assume the likelihoods of the values $\{1,2,3\}$ are $\{\nicefrac{1}{5},\nicefrac{1}{5},\nicefrac{3}{5}\}$ and $\{\nicefrac{3}{5},\nicefrac{1}{5},\nicefrac{1}{5}\}$ for \typeh- and \typel-consumers, respectively.

The platform can provide information only about the consumer's type and not about their value. Hence, the seller's price depends on the likelihood \belief\ the seller attaches to the consumer's type being \typel. Assuming the seller breaks ties in favor of consumers, consumers' welfare as a function of the seller's belief \belief\ and their type \type\ is as follows:
\begin{align}
\welfare(\belief,\typeh)=\left\{\begin{array}{ll}0&\text{ if }\belief<\nicefrac{1}{2}\\
\nicefrac{3}{5}&\text{ if }\belief\in[\nicefrac{1}{2},\nicefrac{3}{4})\\
\nicefrac{7}{5}&\text{ if }\belief\in[\nicefrac{3}{4},1]\end{array}\right.,\;&\welfare(\belief,\typel)=\left\{\begin{array}{ll}0&\text{ if }\belief<\frac{1}{2}\\
\nicefrac{1}{5}&\text{ if }\belief\in[\nicefrac{1}{2},\nicefrac{3}{4})\\
\nicefrac{3}{5}&\text{ if }\belief\in[\nicefrac{3}{4},1]\end{array}\right..
\end{align}
\hfill$\closing$
\end{example}
\color{black}
\begin{example}[Credit ratings]\label{example:credit}
A credit agency makes lending decisions based on an applicant's perceived repayment probability. An applicant of type \type\ repays loans with probability $\linear(\type)$. The credit agency approves loans with probability proportional to the expected value of \linear. A regulator wishes to maximize the probability that applicants of a given type $\type_i$ receive a loan by choosing the information on which the credit agency can condition its approval decision. 
\hfill$\closing$
\end{example}
\color{black}
\label{page-r3-3}\begin{remark}[Interpretation of the welfare function]\label{remark:bp}
The welfare function admits several interpretations. First, following \cite{kage11}, the population's welfare may be determined by the actions taken by the outside observer after observing the realization of an \informationpolicy, as in Examples \ref{example:platforms} and \ref{example:credit}. Using the notation in that paper, denote by $v(\aaction,\type)$ the utility of type-\type\ \agents\ when the outside observer takes action $\aaction\in A$. If the outside observer takes action $\aaction(\belief)$ when her posterior belief is \belief, type-\type\ \agents\ obtain payoff $v(\aaction(\belief),\type)\equiv \welfare(\belief,\type)$.\footnote{Note that the welfare function \welfare\ differs from the sender's indirect utility function in \cite{kage11}, usually denoted by $\hat{v}$. The indirect utility function is the expectation under \belief\ of the welfare function, $\welfare(\belief,\cdot)$. That is, $\hat{v}(\belief)=\sum_{\type\in\Types}\belief(\type)v(\aaction(\belief),\type)=\sum_{\type\in\Types}\belief(\type)\welfare(\belief,\type)$.} Second, the welfare function may also capture that the population's welfare may be driven by image or reputation concerns, as in \cite{beti06} and \cite{tirl20}, or psychological motives, as in \cite{lipnowski2018disclosure}. Indeed, \autoref{example:privacy} admits a psychological interpretation under which an \agent\ derives utility from keeping the outside observer, who is attempting to guess the \agent's type, in suspense (cf. \citealp{ely2015suspense}).
\end{remark}
\color{black}
\section{Characterization}\label{sec:main}
\autoref{sec:main} presents our characterization of the \bset\ via the convex hull of the graph of a vector-valued function, in the spirit of the belief-based approach of \cite{kage11}. 

\paragraph{Truth-drifting:} An apparent obstacle in following the belief approach in \cite{kage11} is that the elements of \repset\ are expressed in terms of expectations \emph{conditional} on a given type $\type\in\types$, rather than unconditional expectations. Indeed, as shown in \cite{frkr14}, conditional expectations do not satisfy the martingale property; rather, they \emph{drift toward the truth}. More precisely, for any type \type\ and for any \informationpolicy\ \policy, the expectation of the posterior probability of \type\ conditional on $\typec=\type$ is higher than the prior probability of \type. That is,\footnote{\autoref{claim:truth-drifting} provides a more general version of this result based on \autoref{theorem:linear-concavification}.}
\begin{align}\label{eq:truth-drifting}\tag{TD}
\mathbb{E}_{\conditional}\left[\frac{\beliefc(\type)}{\prior(\type)}\right]=\sum_{s\in S}\excdf(s|\type)\frac{\belief_s(\type)}{\prior(\type)}\geq1.
\end{align}
 Instead of pursuing a characterization of \emph{conditional} distributions of posteriors, we recover the belief approach in \cite{kage11} by studying a suitably modified welfare function.
\paragraph{Truth-adjusted welfare:} We show any element $\allocv\in\repset$ can be expressed as the unconditional expectation of an adjusted version of the welfare function. Indeed, define the \emph{truth-adjusted} welfare function 
$\Reputation:\Posteriors\times\Types\mapsto\mathbb{R}$ to be
\begin{align}\label{adjusted_payoff_function}\tag{AW}
   \Reputation(\belief,\type)\equiv\frac{\belief(\type)}{\prior(\type)}\welfare(\belief,\type).
\end{align}
That is, \Reputation\ is the welfare function \welfare\ adjusted by the \emph{truth-drift} $\nicefrac{\belief(\type)}{\prior(\type)}$.  For any given posterior belief \belief, the likelihood ratio $\nicefrac{\belief(\type)}{\prior(\type)}$ measures the representation of type \type\ under \belief\ relative to its ex ante representation under \prior. 

The truth-adjusted welfare function combines the preferences of individuals of type \type\ for a particular belief--$\welfare(\belief,\type)$--and the \emph{resource constraint}--Bayes plausibility--in our economy, where \informationpolicies\ take the role of allocations.  Indeed, the likelihood-ratio adjustment $\belief(\type)/\prior(\type)$ captures the constraint that comes from Bayes plausibility. Intuitively, type-\type\ \agents\ have an endowment equal to $\prior(\type)$ that can be spread over different beliefs $\belief(\type)$, and the Bayes plausibility constraint ensures  this spread is done in a way that respects the budget.\footnote{Formally, for any \informationpolicy, \policy,  and type $\type\in\types$, $\mathbb{E}_{\unconditional}[\beliefc(\type)/\prior(\type)]=1$.}


\setcounter{example}{0}
\begin{example}[continued] We illustrate the truth-adjusted welfare function in the context of \autoref{example:privacy}. \autoref{fig:privacy-reputation} depicts the welfare function $\welfare(\belief,\type)$ (\autoref{fig:privacy-welfare}) and the truth-adjusted welfare function for \typea\ (\autoref{fig:privacy-welfare-a}) and \typebb\ (\autoref{fig:privacy-welfare-b}). Whereas in this example welfare is assumed to be type-independent, the truth-adjusted welfare function is type-dependent. This natural consequence of the likelihood-ratio adjustment reflects that \agents\ of different types get to benefit differently from various induced beliefs and therefore, from the same information structure. 
\hfill$\closing$
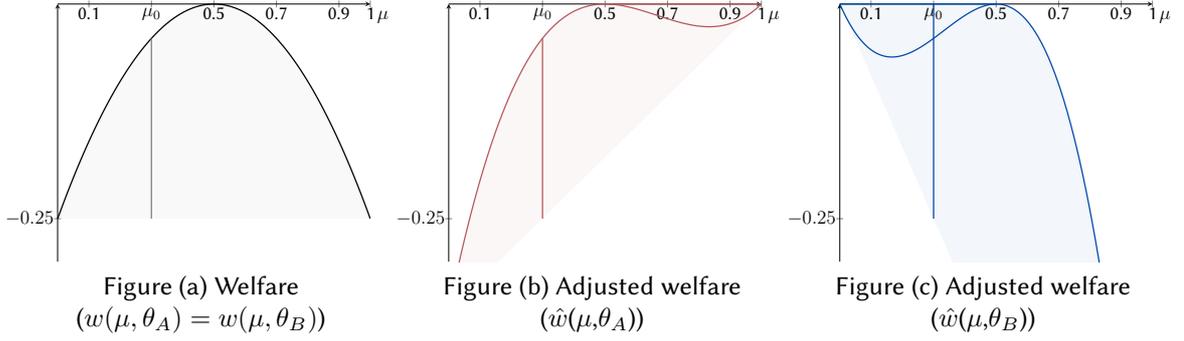
\begin{figure}[t!]
\subfloat[Welfare ($\reputation(\belief,\typea)=\reputation(\belief,\typebb)$)]{\scalebox{0.6}{%
\begin{tikzpicture}
    \begin{axis}[axis lines=middle,
    legend cell align={left},
    legend style={draw=none,font=\footnotesize,anchor=north, at={(1.3,1)}},
            xmin=0,xmax=1,
        ymin=-0.3,ymax=0,
        xtick={0,0.1,0.3,0.5,0.7,0.9,1},
        xticklabels={0,0.1,$\mu_0$,0.5,0.7,0.9,1},
        ytick={-0.25,0},
        xlabel=$\mu$,
        every axis x label/.style={at={(current axis.right of origin)},below right=2mm},
every axis y label/.style={at={(current axis.north west)},left=5mm}]
        \addplot [samples=100,black,thick,domain=0:1,name path=D,mark=$w$]  {-(x-0.5)^2};
             \path[name path=A] (axis cs:0,-0.25) -- (axis cs:1,-0.25);
            \addplot [
        thick,
        color=gray,
        fill=gray, 
        fill opacity=0.05
    ]
    fill between[
        of=D and A,
        soft clip={domain=0:1},
    ];   
   \addplot[thick,draw=gray,mark={}] coordinates {(0.3,-0.25) (0.3,-0.04)};
\end{axis}
\end{tikzpicture}}\label{fig:privacy-welfare}
}
\subfloat[Adjusted welfare (\Reputation(\belief,\typea)) ]{\scalebox{0.6}{%
\begin{tikzpicture}
    \begin{axis}[axis lines=middle,
    legend cell align={left},
    legend style={draw=none,font=\footnotesize,anchor=north, at={(1.3,1)}},
            xmin=0,xmax=1,
        ymin=-0.3,ymax=0,
        xtick={0,0.1,0.3,0.5,0.7,0.9,1},
        xticklabels={0,0.1,$\mu_0$,0.5,0.7,0.9,1},
        ytick={-0.25,0},
        xlabel=$\mu$,
        every axis x label/.style={at={(current axis.right of origin)},below right=2mm},
every axis y label/.style={at={(current axis.north west)},left=5mm}]
                 \addplot [samples=100,sred,thick,domain=0:1,name path=D,mark=$w_A$]  {-(10/7)*(1-x)*((x-0.5)^2)};
                 \path[name path=A] (axis cs:0,-5/14) -- (axis cs:1,0) ;
            \addplot [
        thick,
        color=sred,
        fill=sred, 
        fill opacity=0.05
    ]
    fill between[
        of=D and A,
        soft clip={domain=0:1},
    ];   
\addplot[thick,name path=B, draw=sred,mark={}] coordinates {(0.5,0) (1,0)};
    \addplot [samples=100,sred,thick,domain=0.5:1,name path=E,mark=$w_A$, draw=none]  {-(10/7)*(1-x)*((x-0.5)^2)};
    \addplot [
        thick,
        color=sred,
        fill=sred, 
        fill opacity=0.05
    ]
    fill between[
        of=E and B,
        soft clip={domain=0:1},
    ]; 
   \addplot[thick,draw=sred,mark={}] coordinates {(0.3,-0.25) (0.3,-0.04)};
\end{axis}
\end{tikzpicture}}\label{fig:privacy-welfare-a}
}
\subfloat[Adjusted welfare (\Reputation(\belief,\typebb))]{\scalebox{0.6}{%
\begin{tikzpicture}
    \begin{axis}[axis lines=middle,
    legend cell align={left},
    legend style={draw=none,font=\footnotesize,anchor=north, at={(1.3,1)}},
            xmin=0,xmax=1,
        ymin=-0.3,ymax=0,
        xtick={0,0.1,0.3,0.5,0.7,0.9,1},
        xticklabels={0,0.1,$\mu_0$,0.5,0.7,0.9,1},
        ytick={-0.25,0},
        xlabel=$\mu$,
        every axis x label/.style={at={(current axis.right of origin)},below right=2mm},
every axis y label/.style={at={(current axis.north west)},left=5mm}]
        \addplot [samples=100,sblue,thick,domain=0:1,name path=D,mark=$w_B$]  {-(10/3)*x*((x-0.5)^2)};
                    \path[name path=A] (axis cs:0,0) -- (axis cs:1,-5/6);
    \addplot[thick,name path=B, draw=sblue,mark={}] coordinates {(0,0) (0.5,0)};
            \addplot [
        thick,
        color=sblue,
        fill=sblue, 
        fill opacity=0.05
    ]
    fill between[
        of=D and A,
        soft clip={domain=0:1},
    ];   
    \addplot [
        thick,
        color=sblue,
        fill=sblue, 
        fill opacity=0.05
    ]
    fill between[
        of=D and B,
        soft clip={domain=0:0.5},
    ];   
   \addplot[thick,draw=sblue,mark={}] coordinates {(0.3,-0.25) (0.3,0)};
   
\end{axis}
\end{tikzpicture}}\label{fig:privacy-welfare-b}
}
 \caption{Truth-adjusted welfare function in \autoref{example:privacy}}\label{fig:privacy-reputation}
\end{figure}
\end{example}

\autoref{claim:accounting} justifies our interest in the truth-adjusted welfare function \Reputation:
\begin{claim}[From conditional to unconditional expectations]\label{claim:accounting} For any \informationpolicy\ \policy\ and any type $\type\in\Types$, the following holds:
\begin{align}\label{eq:accounting}
\welfare_\policy(\type)=\expect_{\conditional}[\welfare(\beliefc, \type)]=\mathbbm{E}_{\unconditional}\left[\Reputation(\beliefc,\type)\right].
\end{align}
\end{claim}
The proof of \autoref{claim:accounting} and of other results can be found in the appendix. 

 \autoref{claim:accounting} implies the expectation of \welfare\ under \policy\ \emph{conditional} on type \type\ can be expressed as the \emph{unconditional} expectation of \Reputation\ under \policy. Because the definition of \Reputation\ does not depend on \policy, the distribution over posteriors induced by \policy, \unconditional, is enough to determine the unconditional expectation of \Reputation\ under \policy. Consequently, the analysis that follows relies on the unconditional distribution over posteriors induced by an information structure \policy, rather than on the family of conditional distributions over posteriors induced by \policy.

\color{black}
\label{page-emir-ac}\autoref{claim:accounting} can be obtained from the analysis in \cite{alonso2016bayesian}.\footnote{\cite{rosar2017test} and \cite{quigley2019contradiction} similarly observe that the distribution over posteriors conditional on an \agent's type can be written in terms of the modified unconditional distribution.} Indeed, for a given type \type, one can interpret our model as one of Bayesian persuasion with heterogeneous priors in which the sender assigns probability $1$ to state \type\ and the receiver's prior belief is \prior. \citet[pp.~683--684]{alonso2016bayesian} show the sender's payoff under an information structure \policy\ can be equivalently obtained as the expectation over the distribution of the receiver's posterior beliefs induced by \policy\ of an adjusted payoff function. Specialized to the case in which the sender assigns probability $1$ to state \type\ and the receiver's prior belief is \prior, the truth-adjusted welfare function, $\Reputation(\belief,\type)$, corresponds to the adjusted payoff function in \cite{alonso2016bayesian}.
\color{black}


Relying on \autoref{claim:accounting}, we can immediately characterize the range of welfare values that \agents\ of type \type\ may obtain under some \informationpolicy, via the expected value of $\Reputation(\cdot,\type)$ under a Bayes plausible distribution over posteriors \citep{auma95,kage11}. Indeed, let $\minwelfare(\type),\maxwelfare(\type)$ denote the minimum and maximum welfare individuals of type \type\ can obtain under some information structure. That is, $\minwelfare(\type)=\inf\{\alloc(\type):\allocv\in\repset\}$ and $\maxwelfare(\type)=\sup\{\alloc(\type):\allocv\in\repset\}$. We have the following:
\begin{proposition}[Individually feasible welfare bounds]\label{proposition:bounds}
For any type \type, the following holds:\footnote{For a real-valued function $f$, $\cav f$ denotes the smallest concave function that dominates $f$ and $\vex\;f$ denotes the highest convex function dominated by $f$ \citep{hiriart2004fundamentals}.}
\begin{align*}
\minwelfare(\type)&=\vex\;\Reputation(\prior,\type),\:\maxwelfare(\type)=\cav\;\Reputation(\prior,\type).
\end{align*}
\end{proposition} 
\autoref{proposition:bounds} follows from the main result in \cite{kage11}. The vertical solid lines in Figures \ref{fig:privacy-welfare-a} and \ref{fig:privacy-welfare-b} illustrate the individually feasible welfare values for \typea\ and \typebb\ in \autoref{example:privacy}.  An implication of \autoref{proposition:bounds} is that any \bprofile\ \allocv\ satisfies that for all types \type, $\welfare(\type)\in[\vex\;
\Reputation(\prior,\type),\cav\;\Reputation(\prior,\type)]$.\label{page-r3-7}

Whereas \autoref{proposition:bounds} characterizes what is \emph{individually} feasible for each type in the population, it does not deliver the characterization of the \bset. The reason is that it ignores the across-type restrictions imposed by Bayes' rule. For instance, the truth-adjusted welfare function \eqref{adjusted_payoff_function} highlights that only types on the support of belief \belief\ get to enjoy the payoff of inducing said belief. Similarly, inspection of Figures \ref{fig:privacy-welfare-a} and \ref{fig:privacy-welfare-b} shows \typea's preferred information structure is no disclosure, whereas \typebb\ would prefer \emph{some} disclosure to no disclosure. In other words, the profile $(\maxwelfare(\typea),\maxwelfare(\typebb))$ is not \emph{jointly} feasible.

Instead, the characterization of the \bset\ can be obtained by studying the convex hull of the graph of the \emph{vector-valued} function
$\Reputationv$, $\Reputationv:\Posteriors\mapsto\mathbb{R}^{N}$, where for each $i\in\{1,\dots,N\}$, $\Reputationv_{i}(\belief)\equiv\Reputation(\belief,\type_i)$. Indeed, we have the following:
\begin{theorem}[Belief-based characterization]\label{theorem:concavification} The \bset\ \repset\ satisfies the following:
\begin{align}
    \repset=\left\{\allocv\in\mathbb{R}^{N}:(\priorv,\allocv)\in\mathrm{co}\left(\mathrm{graph\ }\Reputationv\right)\right\}.
    \end{align}
\end{theorem}
\autoref{theorem:concavification} provides a geometric characterization of the set \repset: it is the section at the prior of the convex hull of the graph of the truth-adjusted welfare function $\Reputationv$. Relying on the result in \cite{kage11} that any Bayes plausible distribution over posteriors is the outcome of some \informationpolicy,\footnote{See also \cite{auma95} and \cite{rayo2010optimal}.} \autoref{theorem:concavification} characterizes a more primitive object, the set of welfare profiles that can be generated by some \informationpolicy. Indeed, whereas the main result in \cite{kage11} would allow us to characterize the  (maximal) ex ante welfare a population with welfare function \welfare\ can obtain, \autoref{theorem:concavification} characterizes the welfare profiles whose average leads to that welfare. 
\color{black}
\color{black}
\paragraph{Calculating the \bset:} 
Relying on \autoref{theorem:concavification}, \autoref{fig:privacy-concavification} illustrates the construction of the \bset\ in \autoref{example:privacy}. \autoref{fig:privacy-adjusted-welfare-hull} depicts the convex hull of the graph of $\Reputationv$. Applying \autoref{theorem:concavification}, the resulting \bset\ is the section of this convex hull at $\prior=0.3$. The blue shaded area in \autoref{fig:privacy-adjusted-welfare-hull} depicts this section, which is represented in \autoref{fig:privacy-profileset}.

\autoref{fig:privacy-profileset} illustrates which \informationprofiles\ are jointly feasible under \emph{some} information structure in \autoref{example:privacy}. For instance,  fully revealing or concealing \agents' types is always possible, so that the full and no-disclosure profiles, $\allocv^{FD}$ and $\allocv^{ND}$, are feasible. All \bprofiles\ Pareto dominate the full-disclosure profile, whereas no \bprofile\ Pareto dominates the no-disclosure one. As discussed above, simultaneously giving all \agent\ types their maximum welfare \maxwelfare(\type) is not possible, because \typebb's welfare is maximized by a policy that sometimes reveals an \agent\ is of type $\typea$.

Because the welfare function is continuous in \autoref{example:privacy}, the \bset\ is closed, but this property does not follow from the ongoing assumption that the welfare function is bounded. However, making assumptions other than that the welfare function is bounded may not be natural. To illustrate, consider the case in which the welfare function captures in reduced form that the welfare of the population is determined by the outside observer's actions after observing the realization of the \informationpolicy. As we explained in \autoref{remark:bp}, each selection from the outside observer's best-response correspondence induces \emph{a} welfare function and hence a corresponding \bset. If one assumes--as we do in \autoref{example:platforms}--the outside observer breaks ties in favor of the individuals, one would naturally obtain an upper-semicontinuous welfare function, but under adversarial tie-breaking, having a lower-semicontinuous welfare function would have made sense. As we illustrate in \autoref{sec:pareto}, for a fixed tie-breaking rule, the corresponding \bset\ may not be closed (see, e.g., \autoref{fig:admissions}). When one instead considers the welfare implications of different selection rules, the analogue of the \bset\ is the set of all profiles that are induced by some information structure \emph{and} some selection from the best-response correspondence. As we show in \autoref{proposition:closed-bset} in the appendix, the characterization behind \autoref{theorem:concavification} delivers that this analogue of the \bset\ is closed.\label{page-closed-bset}
\begin{figure}[t!]
    \centering
        \subfloat[The convex hull of the graph of \Reputation]{\scalebox{0.75}{%
        \begin{tikzpicture}
\begin{axis}[xlabel=$\Reputationv_\Left$,ylabel=$\Reputationv_\Right$,zlabel=$\belief$,ztick={0,0.3,1},xtick={-0.25,0},ytick={-0.25,0}, xmin=-0.3,xmax=0.1,ymin=-0.3,ymax=0.1,zmin=-0.1,zmax=1.1,colormap/cool,z label style={rotate=-90},
scale only axis=true, plot box ratio={1}{1}{1},    width=8cm,   height=8cm]

\addplot3+[mark=none,color=blue,thick, fill=cyan] coordinates {(-0.236174,-0.249332, 0.3) (-0.230263,-0.249023, 0.3) (-0.22538,-0.248304, 0.3) (-0.221371,-0.247687, 0.3) (-0.217105,-0.247019, 0.3) (-0.213915,-0.246611, 0.3) (-0.205104,-0.244577, 0.3) (-0.203947,-0.244269, 0.3) (-0.203279,-0.244089, 0.3) (-0.200632,-0.243421, 0.3) (-0.192897,-0.241314, 0.3) (-0.190789,-0.24044, 0.3) (-0.186112,-0.238744, 0.3) (-0.183362,-0.23769, 0.3) (-0.177632,-0.235506, 0.3) (-0.174085,-0.23381, 0.3) (-0.167044,-0.230263, 0.3) (-0.165347,-0.229389, 0.3) (-0.164474,-0.228927, 0.3) (-0.16175,-0.227539, 0.3) (-0.156944,-0.224635, 0.3) (-0.151316,-0.220909, 0.3) (-0.149003,-0.219418, 0.3) (-0.145508,-0.217105, 0.3) (-0.141396,-0.213867, 0.3) (-0.138158,-0.211246, 0.3) (-0.134136,-0.207969, 0.3) (-0.129292,-0.203947, 0.3) (-0.127146,-0.201801, 0.3) (-0.125,-0.199656, 0.3) (-0.120567,-0.195223, 0.3) (-0.116134,-0.190789, 0.3) (-0.114155,-0.188477, 0.3) (-0.111842,-0.185701, 0.3) (-0.108193,-0.181281, 0.3) (-0.10516,-0.177632, 0.3) (-0.102231,-0.174085, 0.3) (-0.0986842,-0.169048, 0.3) (-0.0968082,-0.16635, 0.3) (-0.0954975,-0.164474, 0.3) (-0.0913728,-0.158627, 0.3) (-0.0887644,-0.154554, 0.3) (-0.0868627,-0.151316, 0.3) (-0.0863455,-0.150497, 0.3) (-0.0855263,-0.149106, 0.3) (-0.0814659,-0.142218, 0.3) (-0.0790502,-0.138158, 0.3) (-0.076763,-0.133763, 0.3) (-0.0723684,-0.125051, 0.3) (-0.0723427,-0.125026, 0.3) (-0.0723427,-0.125, 0.3) (-0.072317,-0.124949, 0.3) (-0.0679418,-0.116269, 0.3) (-0.0659886,-0.111842, 0.3) (-0.0639905,-0.107062, 0.3) (-0.0612151,-0.100689, 0.3) (-0.0603927,-0.0986842, 0.3) (-0.0600201,-0.0978747, 0.3) (-0.0592105,-0.0960244, 0.3) (-0.0575529,-0.0921053, 0.3) (-0.0562262,-0.0885106, 0.3) (-0.0552015,-0.0855263, 0.3) (-0.0528115,-0.0791273, 0.3) (-0.0527215,-0.0788574, 0.3) (-0.0503444,-0.0723684, 0.3) (-0.0492136,-0.0692074, 0.3) (-0.0460526,-0.0594675, 0.3) (-0.0460012,-0.0592619, 0.3) (-0.0460012,-0.0592105, 0.3) (-0.0459241,-0.059082, 0.3) (-0.0429174,-0.0491879, 0.3) (-0.0418894,-0.0460526, 0.3) (-0.0400391,-0.0400391, 0.3) (-0.0408614,-0.0394737, 0.3) (-0.0460526,-0.0361328, 0.3) (-0.0480572,-0.0348993, 0.3) (-0.0510896,-0.0328947, 0.3) (-0.0526316,-0.0319696, 0.3) (-0.0561781,-0.0298623, 0.3) (-0.0592105,-0.0280633, 0.3) (-0.0645045,-0.0250308, 0.3) (-0.0702226,-0.0218827, 0.3) (-0.0723684,-0.0207391, 0.3) (-0.0730366,-0.020405, 0.3) (-0.0742059,-0.0197369, 0.3) (-0.0855263,-0.0142887, 0.3) (-0.0909617,-0.0120143, 0.3) (-0.0949322,-0.010331, 0.3) (-0.0986842,-0.00894327, 0.3) (-0.100429,-0.00832328, 0.3) (-0.105263,-0.00668176, 0.3) (-0.105366,-0.00657896, 0.3) (-0.110377,-0.00511412, 0.3) (-0.111842,-0.00467723, 0.3) (-0.114463,-0.00395766, 0.3) (-0.12094,-0.00251852, 0.3) (-0.125,-0.00161905, 0.3) (-0.130705,-0.00087378, 0.3) (-0.132247,-0.000668187, 0.3) (-0.138158,-0.000308401, 0.3) (-0.144647,-0.0000899597, 0.3) (-0.144833,-0.0000963844, 0.3) (-0.151316,-0.000404773, 0.3) (-0.156674,-0.00122072, 0.3) (-0.159591,-0.00169615, 0.3) (-0.164474,-0.0029811, 0.3) (-0.167301,-0.00375207, 0.3) (-0.17465,-0.00657896, 0.3) (-0.176655,-0.00755552, 0.3) (-0.177632,-0.0081209, 0.3) (-0.180664,-0.00961144, 0.3) (-0.184981,-0.0123869, 0.3) (-0.190789,-0.0166016, 0.3) (-0.192486,-0.0180407, 0.3) (-0.194593,-0.0197369, 0.3) (-0.19927,-0.0244141, 0.3) (-0.203947,-0.029798, 0.3) (-0.205335,-0.031507, 0.3) (-0.20644,-0.0328947, 0.3) (-0.210838,-0.0391621, 0.3) (-0.211972,-0.0409193, 0.3) (-0.214947,-0.0460526, 0.3) (-0.215782,-0.0473761, 0.3) (-0.217105,-0.0498561, 0.3) (-0.220241,-0.0560753, 0.3) (-0.22168,-0.0592105, 0.3) (-0.22425,-0.0652241, 0.3) (-0.225175,-0.06728, 0.3) (-0.227025,-0.0723684, 0.3) (-0.227912,-0.0747199, 0.3) (-0.230263,-0.0817229, 0.3) (-0.231137,-0.0846526, 0.3) (-0.231394,-0.0855263, 0.3) (-0.232014,-0.0872771, 0.3) (-0.234118,-0.0948294, 0.3) (-0.235095,-0.0986842, 0.3) (-0.236688,-0.105109, 0.3) (-0.236752,-0.105353, 0.3) (-0.236842,-0.105777, 0.3) (-0.238178,-0.111842, 0.3) (-0.239116,-0.116147, 0.3) (-0.240144,-0.121723, 0.3) (-0.240755,-0.125, 0.3) (-0.241166,-0.127255, 0.3) (-0.242753,-0.13749, 0.3) (-0.242856,-0.138158, 0.3) (-0.242958,-0.13862, 0.3) (-0.243421,-0.141756, 0.3) (-0.243858,-0.144737, 0.3) (-0.2445,-0.150236, 0.3) (-0.244661,-0.151316, 0.3) (-0.244809,-0.152704, 0.3) (-0.245901,-0.161994, 0.3) (-0.246119,-0.164474, 0.3) (-0.246415,-0.167468, 0.3) (-0.247032,-0.174021, 0.3) (-0.24734,-0.177632, 0.3) (-0.247639,-0.181849, 0.3) (-0.247947,-0.186263, 0.3) (-0.248252,-0.190789, 0.3) (-0.248567,-0.195936, 0.3) (-0.248689,-0.198679, 0.3) (-0.248869,-0.203947, 0.3) (-0.249203,-0.20973, 0.3) (-0.249293,-0.211233, 0.3) (-0.249435,-0.217105, 0.3) (-0.249615,-0.223299, 0.3) (-0.24964,-0.224044, 0.3) (-0.24982,-0.230263, 0.3) (-0.249846,-0.236688, 0.3) (-0.249846,-0.236996, 0.3) (-0.249897,-0.243421, 0.3) (-0.250003,-0.249997, 0.3) (-0.250003,-0.250003, 0.3) (-0.249997,-0.250003, 0.3) (-0.249997,-0.25, 0.3) (-0.243421,-0.24964, 0.3) (-0.237446,-0.249396, 0.3) (-0.236174,-0.249332, 0.3)};

\addplot3[
    opacity=0.3,
    table/row sep=\\,
    patch, patch type=polygon, vertex count=3,
    patch table={%
      51  45  46  \\ 51  4  5  \\ 51  1  2  \\ 1  33  32  \\ 1  22  21  \\ 22  51  21  \\ 1  45  44  \\ 1  51  50  \\ 1  8  7  \\ 1  38  37  \\ 51  33  34  \\ 33  51  32  \\ 51  12  13  \\ 12  1  13  \\ 1  42  41  \\ 51  42  43  \\ 1  48  47  \\ 45  51  44  \\ 8  51  7  \\ 4  1  5  \\ 1  4  3  \\ 4  51  3  \\ 38  51  37  \\ 33  1  34  \\ 51  27  28  \\ 27  1  28  \\ 1  12  11  \\ 12  51  11  \\ 51  22  23  \\ 22  1  23  \\ 38  1  39  \\ 1  40  39  \\ 42  1  43  \\ 42  51  41  \\ 51  48  49  \\ 48  1  49  \\ 48  51  47  \\ 45  1  46  \\ 1  36  35  \\ 36  51  35  \\ 51  6  7  \\ 6  1  7  \\ 1  6  5  \\ 6  51  5  \\ 51  2  3  \\ 2  1  3  \\ 51  36  37  \\ 36  1  37  \\ 1  30  29  \\ 30  51  29  \\ 51  30  31  \\ 30  1  31  \\ 51  8  9  \\ 8  1  9  \\ 1  16  15  \\ 16  51  15  \\ 1  27  26  \\ 27  51  26  \\ 51  16  17  \\ 16  1  17  \\ 1  25  24  \\ 25  51  24  \\ 40  51  39  \\ 51  38  39  \\ 51  40  41  \\ 40  1  41  \\ 1  44  43  \\ 44  51  43  \\ 1  50  49  \\ 50  51  49  \\ 51  46  47  \\ 46  1  47  \\ 51  34  35  \\ 34  1  35  \\ 51  28  29  \\ 28  1  29  \\ 1  32  31  \\ 32  51  31  \\ 1  10  9  \\ 10  51  9  \\ 51  10  11  \\ 10  1  11  \\ 51  14  15  \\ 14  1  15  \\ 1  14  13  \\ 14  51  13  \\ 51  25  26  \\ 25  1  26  \\ 1  19  18  \\ 19  51  18  \\ 51  19  20  \\ 19  1  20  \\ 51  23  24  \\ 23  1  24  \\ 51  17  18  \\ 17  1  18  \\ 1  21  20  \\ 21  51  20 \\
    }
]
table[row sep=\\]{
x y z \\
-0.357143  0.  0. \\ 
-0.357143  0.  0.  \\ -0.32256  -0.01536  0.02  \\ -0.290194  -0.028213  0.04  \\ -0.259977  -0.03872  0.06  \\ -0.23184  -0.04704  0.08  \\ -0.205714  -0.053333  0.1  \\ -0.181531  -0.05776  0.12  \\ -0.159223  -0.06048  0.14  \\ -0.13872  -0.061653  0.16  \\ -0.119954  -0.06144  0.18  \\ -0.102857  -0.06  0.2  \\ -0.08736  -0.057493  0.22  \\ -0.073394  -0.05408  0.24  \\ -0.060891  -0.04992  0.26  \\ -0.049783  -0.045173  0.28  \\ -0.04  -0.04  0.3  \\ -0.031474  -0.03456  0.32  \\ -0.024137  -0.029013  0.34  \\ -0.01792  -0.02352  0.36  \\ -0.012754  -0.01824  0.38  \\ -0.008571  -0.013333  0.4  \\ -0.005303  -0.00896  0.42  \\ -0.00288  -0.00528  0.44  \\ -0.001234  -0.002453  0.46  \\ -0.000297  -0.00064  0.48  \\ 0.  0.  0.5  \\ -0.000274  -0.000693  0.52  \\ -0.001051  -0.00288  0.54  \\ -0.002263  -0.00672  0.56  \\ -0.00384  -0.012373  0.58  \\ -0.005714  -0.02  0.6  \\ -0.007817  -0.02976  0.62  \\ -0.01008  -0.041813  0.64  \\ -0.012434  -0.05632  0.66  \\ -0.014811  -0.07344  0.68  \\ -0.017143  -0.093333  0.7  \\ -0.01936  -0.11616  0.72  \\ -0.021394  -0.14208  0.74  \\ -0.023177  -0.171253  0.76  \\ -0.02464  -0.20384  0.78  \\ -0.025714  -0.24  0.8  \\ -0.026331  -0.279893  0.82  \\ -0.026423  -0.32368  0.84  \\ -0.02592  -0.37152  0.86  \\ -0.024754  -0.423573  0.88  \\ -0.022857  -0.48  0.9  \\ -0.02016  -0.54096  0.92  \\ -0.016594  -0.606613  0.94  \\ -0.012091  -0.67712  0.96  \\ -0.006583  -0.75264  0.98  \\ 0.  -0.833333  1. \\
};

 \addplot3+[mark=none,color=\plotcolor,thick, smooth] coordinates {(-(5/14), 0, 0) (-(1539/5600), -(27/800), 1/20) (-(36/175), -(4/75), 1/10) (-(119/800), -(49/800), 3/20) (-(18/175), -(3/50), 1/5) (-(15/224), -(5/96), 1/4) (-(1/25), -(1/25), 3/10) (-(117/5600), -(21/800), 7/20) (-(3/350), -(1/75), 2/5) (-(11/5600), -(3/800), 9/20) (0, 0, 1/2) (-(9/5600), -(11/2400), 11/20) (-(1/175), -(1/50), 3/5) (-(9/800), -(39/800), 13/20) (-(3/175), -(7/75), 7/10) (-(5/224), -(5/32), 3/4) (-(9/350), -(6/25), 4/5) (-(21/800), -(833/2400), 17/20) (-(4/175), -(12/25), 9/10) (-(81/5600), -(513/800), 19/20) (0, -(5/6), 1)};
\end{axis}
\end{tikzpicture}}\label{fig:privacy-adjusted-welfare-hull}
}
\subfloat[The \bset]{
\scalebox{0.8}{%
\begin{tikzpicture}
    \begin{axis}[ xtick pos=left,
    ytick pos=left,
    xtick={-0.25,0},
        ytick={-0.25,0},
    xmin=-0.3,xmax=0.05,
    ymin=-0.3,ymax=0.05,
        xlabel=$\allocv_\Left$,ylabel=$\allocv_\Right$,
         x label style={at={(axis description cs:1,-0.01)}},
    y label style={at={(axis description cs:0.1,1)},rotate=-90}
    ,
     width=8cm,   height=8cm
    ]
   \addplot[color=red,mark=*]coordinates{(-0.04,-0.04)}
[every node/.style={xshift=15pt}]
node[pos=0] { $\color{black}{\allocv^{ND}}$};
\addplot[color=blue,mark=*]coordinates{(-0.25,-0.25)}
[every node/.style={yshift=11pt,xshift=-10pt}]
node[pos=0] {$\color{black}{\allocv^{FD}}$};
\addplot+[mark=none,thick,draw=blue,fill=blue,opacity=0.2] coordinates{(-0.236335,-0.249492) (-0.230263,-0.249023) (-0.225278,-0.248407) (-0.22168,-0.247995) (-0.217105,-0.247224) (-0.213867,-0.246659) (-0.205335,-0.244809) (-0.203947,-0.2445) (-0.203074,-0.244295) (-0.199938,-0.243421) (-0.192897,-0.241314) (-0.190789,-0.240594) (-0.186485,-0.239116) (-0.183234,-0.237819) (-0.177632,-0.235506) (-0.174008,-0.233887) (-0.166632,-0.230263) (-0.165245,-0.229492) (-0.164474,-0.229081) (-0.161929,-0.227719) (-0.15686,-0.224719) (-0.151316,-0.221224) (-0.1489,-0.219521) (-0.145456,-0.217105) (-0.141293,-0.21397) (-0.138158,-0.211451) (-0.134033,-0.208072) (-0.129266,-0.203947) (-0.127107,-0.20184) (-0.125,-0.199861) (-0.120451,-0.195338) (-0.116108,-0.190789) (-0.114097,-0.188534) (-0.111842,-0.185907) (-0.108026,-0.181448) (-0.105006,-0.177632) (-0.102231,-0.174085) (-0.0986842,-0.169305) (-0.0966797,-0.166478) (-0.0952919,-0.164474) (-0.09136,-0.15864) (-0.0885974,-0.154387) (-0.0867085,-0.151316) (-0.0862973,-0.150545) (-0.0855263,-0.149363) (-0.0813888,-0.142295) (-0.0790502,-0.138158) (-0.0767373,-0.133789) (-0.0723684,-0.125463) (-0.0722399,-0.125128) (-0.07215,-0.125) (-0.0719315,-0.124563) (-0.0679353,-0.116275) (-0.0659437,-0.111842) (-0.0638878,-0.107165) (-0.0610095,-0.100483) (-0.0603413,-0.0986842) (-0.0599815,-0.0979132) (-0.0592105,-0.0961143) (-0.0575401,-0.0921053) (-0.0562037,-0.0885331) (-0.0550473,-0.0855263) (-0.0526316,-0.0790502) (-0.0526059,-0.0789731) (-0.0525802,-0.078896) (-0.0502673,-0.0723684) (-0.0491879,-0.0692331) (-0.0460526,-0.0597245) (-0.0459241,-0.059339) (-0.0459498,-0.0592105) (-0.0457956,-0.0589535) (-0.0428146,-0.0492907) (-0.041838,-0.0460526) (-0.0400391,-0.0400391) (-0.04081,-0.0394737) (-0.0460526,-0.0360814) (-0.0479929,-0.034835) (-0.0510896,-0.0328947) (-0.0526316,-0.0319696) (-0.0561202,-0.0298044) (-0.0592105,-0.0280633) (-0.0644499,-0.0249762) (-0.0703382,-0.0217671) (-0.0723684,-0.0207134) (-0.0730366,-0.020405) (-0.0742059,-0.0197369) (-0.0855263,-0.0142887) (-0.0909231,-0.0119758) (-0.0949322,-0.010331) (-0.0986842,-0.00894327) (-0.10038,-0.00827509) (-0.105263,-0.00668176) (-0.105366,-0.00657896) (-0.110352,-0.00508842) (-0.111842,-0.00462584) (-0.114528,-0.00389341) (-0.120837,-0.00241572) (-0.125,-0.00161905) (-0.130808,-0.000770983) (-0.132193,-0.000613577) (-0.138158,-0.000154207) (-0.144685,-0.0000514112) (-0.144788,-0.0000514112) (-0.151316,-0.000404773) (-0.156713,-0.00118217) (-0.159591,-0.00169615) (-0.164474,-0.00286545) (-0.167301,-0.00375207) (-0.174766,-0.00657896) (-0.17671,-0.00750091) (-0.177632,-0.00796671) (-0.1806,-0.0095472) (-0.185059,-0.0123098) (-0.190789,-0.0165502) (-0.19255,-0.0179765) (-0.194593,-0.0197369) (-0.19927,-0.0244141) (-0.203947,-0.0297081) (-0.205335,-0.031507) (-0.206466,-0.0328947) (-0.210848,-0.0391525) (-0.212017,-0.0409642) (-0.215062,-0.0460526) (-0.215795,-0.0473633) (-0.217105,-0.0497469) (-0.220241,-0.0560753) (-0.221731,-0.0592105) (-0.224307,-0.0651663) (-0.225175,-0.06728) (-0.227076,-0.0723684) (-0.22795,-0.0746813) (-0.230263,-0.0815173) (-0.231214,-0.0845755) (-0.231522,-0.0855263) (-0.232036,-0.0872996) (-0.234118,-0.0948294) (-0.235095,-0.0986842) (-0.236752,-0.105173) (-0.236765,-0.10534) (-0.236842,-0.105572) (-0.238178,-0.111842) (-0.239129,-0.116134) (-0.240183,-0.121762) (-0.240774,-0.125) (-0.241185,-0.127236) (-0.242856,-0.137593) (-0.242907,-0.138158) (-0.242958,-0.13862) (-0.243421,-0.14155) (-0.243858,-0.144737) (-0.244565,-0.150172) (-0.244706,-0.151316) (-0.244886,-0.152781) (-0.245901,-0.161994) (-0.246145,-0.164474) (-0.246454,-0.167506) (-0.247045,-0.174008) (-0.24734,-0.177632) (-0.247661,-0.181872) (-0.24797,-0.186241) (-0.248252,-0.190789) (-0.248574,-0.195942) (-0.248715,-0.198653) (-0.248985,-0.203947) (-0.249229,-0.209755) (-0.249293,-0.211233) (-0.249486,-0.217105) (-0.24964,-0.223324) (-0.24964,-0.224044) (-0.24982,-0.230263) (-0.24991,-0.236752) (-0.24991,-0.236932) (-0.249897,-0.243421) (-0.250003,-0.249997) (-0.250003,-0.250003) (-0.249997,-0.250003) (-0.249997,-0.25) (-0.243421,-0.24982) (-0.237305,-0.249537) (-0.236335,-0.249492)
};
\addplot[ultra thin,dashed]{x};
\end{axis}
\end{tikzpicture}}
        
        \label{fig:privacy-profileset}}
\caption{Constructing the \bset\ in \autoref{example:privacy}}
\label{fig:privacy-concavification}
\end{figure}

\paragraph{Cardinality:} \autoref{theorem:concavification} has an immediate implication for the cardinality of the \informationpolicies\ that generate points in \repset:
\begin{corollary}\label{corollary:cardinality}
Let $\allocv\in\repset$. Then, an \informationpolicy\ \policy\ with at most $2N$ signals exists such that $\allocv_i=\reputation_{\policy}(\type_i)$ for all $i\in\{1,\dots,N\}$.
\end{corollary}
\autoref{corollary:cardinality} stands in contrast to the result in Bayesian persuasion that  finding an \informationpolicy\ that delivers the sender's maximal payoff and employs at most $N$ posteriors is always possible. There are two reasons behind this difference. First, \autoref{theorem:concavification} characterizes \emph{all} \informationprofiles\ consistent with some information structure so that without further assumptions on the truth-adjusted welfare function, we rely on Carath\'edory's theorem to obtain \autoref{corollary:cardinality}. Instead, \cite{kage11} characterize the sender's maximal payoff and assume the indirect utility function is upper-semicontinuous. For that reason, \cite{kage11} can rely on F\'enchel-Bunt's theorem instead of Carath\'eodory to obtain an upper bound of $N$ instead of $N+1$. Second, we are interested not just in the payoff of one ``sender,'' but in the payoff of $N$, one for each type.

In \autoref{example:privacy}, the upper bound in \autoref{corollary:cardinality} is loose. Because the truth-adjusted welfare function is continuous, we can rely on F\'enchel-Bunt's theorem to reduce the upper bound in \autoref{corollary:cardinality} by 1. Furthermore, the \bprofiles\ on the boundary of \repset\ are induced by \informationpolicies\ with at most two signals. We explain the reason for this further reduction in \autoref{sec:pareto}, where we characterize the boundary of the \bset\ and, in particular, its Pareto frontier.

\section{The Pareto frontier of the \bset}\label{sec:pareto}
We characterize in this section the Pareto frontier of \repset. Points in the Pareto frontier are natural candidates for being the outcome of efficient bargaining over \informationpolicies\ or a social planner's choice. Indeed, these points correspond to the solution of a utilitarian planner as we vary the weights the planner assigns to different types. \autoref{theorem:hyperplanes} shows these points can be recovered as solutions to \emph{standard} Bayesian persuasion problems, where an information designer takes the role of the utilitarian planner. Armed with this characterization, \autoref{proposition:pareto-improvement} provides a necessary and sufficient condition for the no-disclosure profile to be part of the Pareto frontier. In other words, it provides a necessary and sufficient condition for information disclosure to lead to a Pareto improvement relative to no disclosure. 

We define the Pareto frontier of \repset\ to be the set of weak Pareto efficient \bprofiles. Formally, 
\begin{align}\label{eq:Pareto}\tag{P}
\Pareto=\{\allocv\in\repset:(\nexists\allocbv\in\repset)\allocbv>\allocv\}.
\end{align}
Because the \bset\ \repset\ is convex, for any $\allocv\in\Pareto$, the separating hyperplane theorem implies a direction $\directionv\in\reals_+^N\setminus\{0\}$ exists such that\footnote{The maximum is attained because by definition $\welfarev\in\repset$. As we show in \autoref{appendix:pareto}, that $\directionv\in\reals_+^N\setminus\{0\}$ follows from $\welfarev\in\Pareto$.} 
 \begin{align}\label{eq:hyperplanes}
\directionv^T\allocv=\max\left\{\directionv^T\allocbv:\allocbv\in\repset\right\}&=\max\left\{\directionv^T\mathbb{E}_\bsplit\left[\Reputationv(\beliefc)\right]:\bsplit\in\bplausible\right\}\nonumber\\
&=\max\left\{\mathbb{E}_\bsplit\left[\directionv^T\Reputationv(\beliefc)\right]:\bsplit\in\bplausible\right\}.
\end{align}
The first equality simply states that $\allocv$ is a maximizer of the support function of the \bset\ in direction \directionv. Instead, the second equality follows from \autoref{theorem:concavification}. Indeed, \autoref{theorem:concavification} implies we can exchange the maximization over \informationprofiles\ in \repset\ for a maximization over Bayes plausible distributions over posteriors. Note we can interchangeably talk about Pareto efficient \bprofiles\ and Pareto efficient \informationpolicies, and we do this in what follows.

Once we note we can restrict attention to directions $\directionv\in\Posteriors$, \autoref{eq:hyperplanes} has two economic interpretations. First, consider the problem of a social planner who assigns weight \direction(\type) to type \type\ and wishes to maximize the weighted sum of utilities of each type. Under this interpretation, \autoref{eq:hyperplanes} states that \allocv\ is a solution to the social planner's problem. Second, we can interpret $\directionv^T\allocv$ as the expectation with respect to \type\ of the welfare profile \allocv\ under the measure \directionv. In this case, \autoref{eq:hyperplanes} implies $\allocv$ is the vector of \emph{interim} payoffs of a sender with payoff function $\reputation(\belief,\type)$ and prior \directionv. For instance, when $\directionv=\priorv$, so that the sender's prior coincides with that of the outside observer, the above problem coincides with that of \cite{kage11}. Instead, whenever the direction \directionv\ is any element of \Posteriors, the above problem coincides with that considered by \cite{alonso2016bayesian}.\footnote{Thus, one can always interpret the heterogeneous priors model in \cite{alonso2016bayesian} as a model in which the sender and the receiver share the same prior, but the sender assigns weights different than those under the prior \prior\ to each of his possible types.} 

Moreover, \autoref{eq:hyperplanes} has an important practical implication: any Pareto efficient \bprofile\ is induced by the solution to a \emph{supporting Bayesian persuasion problem}. A supporting Bayesian persuasion problem is an instance of the model in \cite{kage11} in which the sender's indirect utility function equals
\begin{align}\label{eq:indirect-direction}
\hat v_{\lambda}(\belief)=\directionv^T\Reputationv(\belief)=\sum_{\type\in\Types}\direction(\type)\frac{\belief(\type)}{\prior(\type)}\welfare(\belief,\type)=\sum_{\type\in\Types}\belief(\type)\frac{\direction(\type)}{\prior(\type)}\welfare(\belief,\type).
\end{align}\label{page-r3-8}
This indirect utility is the product of the truth-adjusted welfare function and the Pareto weights, which capture the rate of substitution between the truth-adjusted welfare of different types. \autoref{theorem:hyperplanes} summarizes the above discussion:
\begin{theorem}[Pareto frontier]\label{theorem:hyperplanes}
The welfare profile $\allocv\in\reals^N$ is in the Pareto frontier of \repset\ if and only if a direction $\directionv\in\Posteriors$ exists such that $\allocv$ is the profile induced by an \informationpolicy\ that solves the supporting Bayesian persuasion problem with indirect utility function $\hat v_\direction$. 
\end{theorem}
Because \repset\ is convex, the characterization in \autoref{theorem:hyperplanes} extends to all \bprofiles\ on the boundary of \repset, except that the supporting direction may no longer be in \Posteriors. 

\color{black}
\autoref{theorem:hyperplanes} characterizes the Pareto frontier of \repset\ by connecting the solution of the utilitarian planner with weights \direction\ to the Bayesian persuasion problem of the sender with indirect utility function $\hat{v}_\direction$. Whereas \autoref{theorem:concavification} characterizes the \bset\ via the convex hull of the graph of a vector-valued function, $\Reputationv$, the results in \cite{kage11} imply the supporting Bayesian persuasion problems corresponding to Pareto efficient profiles can be solved by concavifying a real-valued function, $\hat{v}_\direction$. \autoref{theorem:hyperplanes} thus provides us with a tractable way of characterizing the Pareto efficient \informationpolicies\ and thus recover the Pareto efficient profiles (see, e.g., the analysis in \autoref{sec:linear}). The ability to recover the Pareto frontier may be useful when maximizing non-utilitarian social welfare functions, such as those that correspond to a fairness-aware planner (e.g., Rawls' criterion or \citeauthor{epstein1992quadratic}'s quadratic social welfare function), or those that obtain from efficient bargaining (e.g., Nash bargaining). Even if such an objective would select a Pareto efficient profile, the supporting direction \directionv\ may only be identifiable after characterizing the solution, so that knowledge of the whole Pareto frontier may be important.\label{page-r3-5}

\color{black}

\autoref{theorem:hyperplanes} has yet another practical implication: when a (Pareto efficient) \allocv\ is an extreme point of \repset, an \informationpolicy\ exists that employs at most $N$ signals and generates \allocv. We record this result below:
\begin{corollary}[Extreme points of \repset]\label{corollary:cardinality-2}
Let $\allocv$ be an extreme point of $\repset$. Then, an \informationpolicy\ \policy\ with at most $N$ signals exists such that $\allocv_i=\welfare_\policy(\type_i)$ for all $i\in\{1,\dots,N\}$.
\end{corollary}
In \autoref{example:privacy}, all points in $\Pareto$ are extreme (\autoref{fig:privacy-profileset}). \autoref{corollary:cardinality-2} then implies \informationpolicies\ that induce at most two posteriors are enough to characterize the Pareto frontier in \autoref{example:privacy}.

We use \autoref{example:platforms} to illustrate properties of the \bprofiles\ on the Pareto frontier and why the bound in \autoref{corollary:cardinality-2} applies only to \emph{extreme} points on the Pareto frontier of \repset. In doing so, we highlight the difference between the value of the program in \autoref{eq:hyperplanes} and the \bprofiles\ consistent with that value:
\setcounter{example}{1}
\begin{example}[continued] 
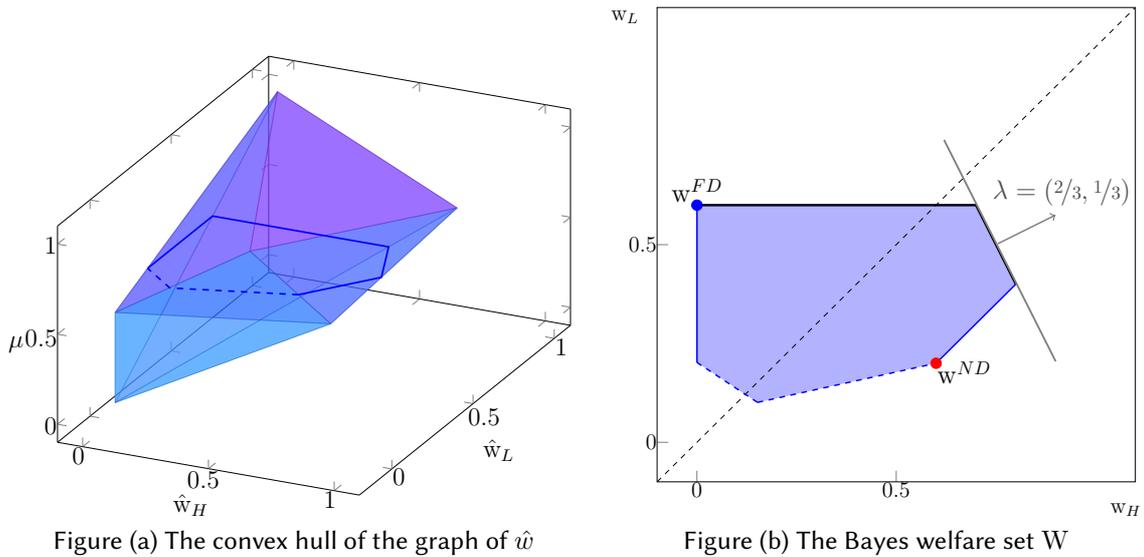
\begin{figure}[t!]
        \subfloat[The convex hull of the graph of \Reputation]{\scalebox{0.8}{%
        \begin{tikzpicture}[scale=1]
\begin{axis}[xlabel=$\hat{\allocv}_H$,ylabel=$\hat{\allocv}_L$,zlabel=$\mu$,ztick={0,0.5,1},xtick={0,0.5,1},ytick={0,1/2,1}, xmin=-0.1,xmax=1.1,ymin=-0.2,ymax=1.1,zmin=-0.1,zmax=1.1,colormap/cool,z label style={rotate=-90},
scale only axis=true, plot box ratio={1}{3/2}{1}]
\addplot3[
    opacity=0.5,
    table/row sep=\\,
    patch, patch type=polygon, vertex count=3,
    patch table={%
        1 0 2\\
        0 1 5\\
           4 2 0\\
    5 4 0\\
        3 1 2\\
        1 3 5\\
        4 3 2\\
        3 4 5\\
    }
]
table[row sep=\\]{
x y z \\
0 0 0\\
0 0 0.5\\
0.75 0.166 0.5\\
0.375 0.25 0.75\\ 
0.875 0.75 0.75\\ 
0 1 1\\
};

\addplot3+[mark=none,color=blue,thick] coordinates {(0,0.2,6/10) (0,0.6, 6/10) (0.7,0.6,6/10) (0.8, 0.4,6/10) (0.6,0.2,6/10)};
\addplot3+[mark=none,color=blue,dashed,thick] coordinates {(0.6,0.2,6/10) (0.152,0.101,6/10) (0,0.2,6/10)};
\end{axis}
\end{tikzpicture}}\label{fig:platforms-adjusted-welfare-hull}
}
\hspace{0.25cm}
\subfloat[The \bset\  \repset]{
\scalebox{0.7}{%
\begin{tikzpicture}[scale=1.4]
\begin{axis}[
            xmin=-0.1,xmax=1.1,
        ymin=-0.1,ymax=1.1,
        xtick={0,0.5},
        ytick={0,0.5},
        xlabel=$\allocv_H$,ylabel=$\allocv_L$
         ,
         x label style={at={(0.7,0.01)}},
    y label style={at={(0.05,0.7)},rotate=-90},
      xtick pos=left,
    ytick pos=left
    ,
    width=8cm,height=8cm
]
\addplot[ultra thin,dashed,domain=-0.1:1.1] {x};
\addplot[color=blue,mark=*]coordinates{(0,0.6)}
[every node/.style={yshift=8pt}]
node[pos=0] {\large$\color{black}{\allocv^{FD}}$};

\addplot[color=red,mark=*]coordinates{(0.6, 0.2)}
[every node/.style={xshift=15pt,yshift=-5pt}]
node[pos=0] {\large$\color{black}{\allocv^{ND}}$};

\addplot +[name path=A,mark=none,color=blue, thick] coordinates {(0,0.6) (0.7, 0.6)};
\addplot +[mark=none,color=black, very thick,solid] coordinates {(0.7, 0.6) (0.8, 0.4)};
\addplot +[mark=none,color=black, very thick,solid] coordinates {(0.7, 0.6) (0,0.6)};
\addplot +[name path=C,mark=none,color=blue, solid,thick] coordinates {(0.6,0.2) (0.8, 0.4)};
\addplot +[solid,name path=D, mark=none,color=blue,dashed, thick] coordinates {(0.152,0.101) (0.6,0.2)};
\addplot +[solid,mark=none,color=blue, thick] coordinates {(0,0.6) (0, 0.202)};
\addplot +[name path=B,dashed,mark=none,color=blue, thick] coordinates {(0.152,0.101) (0, 0.202)};
\addplot [blue!30] fill between [of = A and B, soft clip={domain=0:1}];
\addplot [blue!30] fill between [of = A and C, soft clip={domain=0:1}];
\addplot [blue!30] fill between [of = A and D, soft clip={domain=0:1}];

 \addplot [domain=0.62:0.9, color=white!50!black, line width=1pt,solid] {-2*x+2+0.005};
\addplot [->,color=white!50!black, domain=15/20+0.005:0.9,thick] {0.5*x+5/40};
\addplot[color=white!30!black,mark=none]coordinates{(1,0.6)}
[every node/.style={xshift=-18pt,yshift=7pt}]
node[pos=0] {\large$\direction=(\nicefrac{2}{3},\nicefrac{1}{3})$};
\end{axis}
\end{tikzpicture}}    
        \label{fig:admissions}}
        
        \caption{Constructing the \bset\ in \autoref{example:platforms}. The diagonal dashed line is the $45^\circ$ degree line. Blue dashed lines denote profiles on the boundary of the \bset, which are not in the \bset.}\label{fig:platforms}
\end{figure}
\autoref{fig:platforms} illustrates the convex hull of the graph of $\Reputationv$ (\autoref{fig:platforms-adjusted-welfare-hull}) and the \bset\ (\autoref{fig:admissions}) for the online marketplace example. We note the following features of the Pareto frontier of \repset\ in this example, depicted in black in \autoref{fig:admissions}. First, in contrast to \autoref{example:privacy}, the full-disclosure profile, \allocvF, is Pareto efficient, whereas the no-disclosure profile, \allocvN, is not. Second, like in \autoref{example:privacy}, there is a continuum of \bprofiles\ that are fair in the sense of equalizing welfare across different consumer types. Whereas the points in the flat segment of the Pareto frontier to the right of $\welfarev=(0.6,0.6)$ maximize Rawls' criterion, only $\welfarev=(0.6,0.6)$ is both Pareto efficient and fair.

Third, contrary to \autoref{example:privacy}, not every point on the boundary of \repset\ is an extreme point. Consider, for instance, the points on the boundary in the direction $\direction=(\nicefrac{2}{3},\nicefrac{1}{3})$ in \autoref{fig:admissions}, which would be consistent with the platform being interested in promoting the participation of \typeh-consumers. It is possible to show that the points in the interior of this line are generated by \informationpolicies\ with three signals.\footnote{For instance, the midpoint $\allocv=(3/4,1/2)$ can only be generated by an \informationpolicy\ that employs at least three signals. One such \informationpolicy\ is given by
\medskip
\begin{center}
\begin{tabular}{c|ccc}
$\type_H$& 1/4 & 3/8 & 3/8\\
 \typel\ & 0 & 1/4 & 3/4
\end{tabular}.
\end{center}
 } Note  all the points in that line lead to the same value in the Bayesian persuasion problem with indirect utility function $\hat{v}_{\direction}$, namely, $\cav\;\hat{v}_{\direction}(\prior)$. However, they correspond to different welfare profiles with different implications regarding how \typeh- and \typel-consumers share the payoff $\cav\;\hat{v}_{\direction}(\prior)$. This highlights a benefit of the perspective we develop in this paper: insofar as one cares about the cross-sectional implications of different \informationpolicies, studying only the average welfare induced by a given \informationpolicy\ may not be sufficient.
 
  \label{page-r3-6}
Because the seller breaks ties in favor of the consumer when setting prices, the welfare function in \autoref{example:platforms} is upper-semicontinuous, but not continuous. Consequently, the \bset\ is not closed in this example: The profiles corresponding to the blue dashed lines in \autoref{fig:admissions} can be induced by an \informationpolicy\ only when the seller breaks ties against the consumer. Note, however, that upper-semicontinuity of the welfare function guarantees that the supporting Bayesian persuasion problem attains a solution in all directions $\directionv\in\Posteriors$. Whereas in general upper-semicontinuity does not guarantee the Pareto frontier is closed, it ensures that all weakly monotone social welfare functions attain a maximum in the \bset\ (see \autoref{proposition:usc-general}).
\hfill$\closing$
 \end{example}
We close this section by noting that when \repset\ is closed, \autoref{theorem:hyperplanes} provides an alternative characterization of the \bset. Whereas points on the Pareto frontier are natural candidates for the choice of a social planner, points outside the Pareto frontier may be relevant if, for instance, the social planner faces constraints in their choice of a \bprofile. Because such constraints may rule out Pareto efficient and even boundary profiles, knowledge of the whole \bset\ may be important.  For this reason, we see the characterizations in both theorems as complementary.

%
%
%
%
\subsection{When is (no) information disclosure efficient?}\label{sec:disclosure} 
A natural question to ask is when providing some information is the efficient thing to do; in other words, when does an \informationpolicy\ exist that all types strictly prefer to no disclosure? When such an \informationpolicy\ exists, we say all types \emph{benefit from disclosure}. Examples \ref{example:privacy} and \ref{example:platforms} offer an interesting contrast in this respect: because the no-disclosure profile \allocvN\ is not in the Pareto frontier in \autoref{example:platforms}, all consumer types benefit from disclosure. Instead, the no-disclosure profile \allocvN\ is in the Pareto frontier in \autoref{example:privacy}. Indeed, in \autoref{example:privacy} information disclosure necessarily hurts at least one of the types (in this case, \typea). 


\autoref{proposition:pareto-improvement} provides a necessary and sufficient condition for all types to benefit from disclosure. To introduce this condition,  recalling a definition from \cite{kage11} is useful: a sender with indirect utility function $\hat{v}$ \emph{benefits from persuasion} if the concavification of $\hat{v}$ at the prior exceeds the value of $\hat{v}$ at the prior; that is, $\hat{v}(\prior)<\cav\;\hat{v}(\prior)$. We have the following:

\begin{corollary}[Efficiency of disclosure]\label{proposition:pareto-improvement}
All types benefit from disclosure if and only if, for all  directions $\direction\in\Posteriors$, a sender with indirect utility $\hat{v}_\direction$ benefits from persuasion.
\end{corollary}
In other words, \autoref{proposition:pareto-improvement} states that the no-disclosure profile $\allocvN$ is in the Pareto frontier of \repset\ if and only if a direction $\direction^{ND}\in\Posteriors$ exists such that a sender with indirect utility $\hat{v}_{\direction^{ND}}$ does not benefit from persuasion. Alternatively, a social planner with Pareto weights $\direction^{ND}$ would find no disclosure to be an optimal \informationpolicy. Recall \autoref{example:privacy}: in that case, when $\direction=\prior$, the indirect utility function $\hat{v}_{\prior}$ equals the strictly concave function $-(\belief-\nicefrac{1}{2})^2$, so that no disclosure is optimal. \autoref{corollary:no-benefits} in \autoref{appendix:pareto} shows that this observation reflects a more general result: if for all $\type\in\Types$ the welfare function takes the form $a(\type)\welfare(\belief)+b(\type)$, where $a(\type)>0$ and $\welfare$ is concave, $\allocvN$ is in the Pareto frontier of \repset.




Whereas \autoref{proposition:pareto-improvement} provides a condition in terms of the indirect utility function $\hat{v}_\direction$, \autoref{corollary:benefits-new} provides conditions on the truth-adjusted welfare function \Reputation\ under which all types do or do not benefit from disclosure:\footnote{\autoref{appendix:pareto} provides weaker conditions under which all types (do not) benefit from disclosure.}
\begin{observation}[Disclosure benefits]\label{corollary:benefits-new}
The following hold:
\begin{enumerate}
    \item If, for all $\type\in\Types$, $\Reputation(\cdot,\type)$ is strictly convex in a neighborhood of the prior, all types benefit from disclosure.
    \item Instead, if a type \type\ exists such that $\Reputation(\cdot,\type)$ is concave in \belief, no disclosure is Pareto efficient.
\end{enumerate}
\end{observation}


As in \cite{kage11}, concavity and convexity properties of a payoff function determine whether information disclosure is beneficial. In contrast to \cite{kage11}, the concavity and convexity properties of the truth-adjusted welfare function are what determine whether any given type can benefit from disclosure. This result can be clearly seen in \autoref{example:privacy}, where the truth-adjusted welfare of \typea\ is concave around the prior and that of \typebb\ is strictly convex around the prior, even though each type's welfare function is strictly concave.\footnote{For an even starker example,  consider  $\types=\{\type_1,\type_2\}$ and $\welfare(\belief,\type_1)=-\belief^2-2\belief+3$, $\welfare(\belief,\type_2)=-\belief^2+4\belief$, where $\mu\equiv\mu(\type_2)$. For each type, the welfare function is concave; however, for any prior distribution, the truth-adjusted welfare functions are globally convex. As a result, the unique Pareto efficient \informationpolicy\ is full disclosure.} 

To understand the role of the convexity (concavity) properties of the truth-adjusted welfare function in determining the benefits from disclosure, note that information disclosure affects the welfare of a given type \type\ through two channels: directly through its impact on the welfare function as in \cite{kage11} and indirectly through the truth-drift adjustment. The second channel is most easily seen in the case of binary types. In that case, the property in \autoref{eq:truth-drifting} ensures that the posterior belief drifts along a straight line toward the true type \type. When the welfare function of \agents\ of type \type\ is convex \emph{and} increasing in $\belief(\type)$, both effects are positive, ensuring that \agents\ of type \type\ benefit from disclosure. \autoref{corollary:benefits-binary} below summarizes this discussion:
\begin{observation}[Binary types]\label{corollary:benefits-binary}
Let $N=2$. If  for all  $\type\in\types$, $\welfare(\belief,\type)$ is  strictly convex in $\belief$ and increasing in $\belief(\type)$, both types benefit from disclosure; moreover, full disclosure is uniquely Pareto efficient. Instead, if for some  $\type\in\types$, $\welfare(\belief,\type)$ is  concave in $\belief$ and decreasing in $\belief(\type)$,  no disclosure is Pareto efficient.
\end{observation}

\section{Expected Reputation}\label{sec:linear}
We now specialize our results to the case in which the welfare function equals the expectation of some one-dimensional variable of interest, such as an \agent's productivity, quality, or trade value. This is a standard way to model reputation, image, or career concerns in economics (see, e.g., \citealp{holm99}, \citealp{beti06}). It also captures the class of prediction policy problems in \cite{rambachan2020economic}, in which a decision maker--our outside observer--selects a treatment (e.g., hiring, bail, loan-approval) on the basis of the prediction of an outcome of interest (e.g., productivity, recidivism, creditworthiness). 

Formally, we assume a \emph{reputation vector} $\linearv\in\mathbb{R}^N$ exists such that for all $\type_i,\type_j\in\types$,
\begin{align}\label{eq:linear}\welfare(\belief,\type_i)=\welfare(\belief,\type_j)=\expect_\belief[\rho(\type)]=\sum_{k=1}^N\belief(\type_k)\rho(\type_k)=\beliefv^T\linearv.
\end{align}
We refer to \welfare\ as the \agent's \emph{reputation}. Thus, a \bprofile\ is a profile of \emph{expected} reputations. Without loss of generality, we label $\linearv$ in increasing order, that is, $\linearv_{1}\leq\dots\leq\linearv_{N}$, so that types are labeled in increasing order of their values under \linear.

The analysis in this section allows us to focus on the \emph{redistributive} role of information. Indeed, all information structures lead to the same ex ante welfare. That is, for any \informationpolicy\ \policy\ and the corresponding \bprofile\ \allocv, the ex ante welfare is:
%
\begin{align}\label{eq:linear-total-welfare}
\priorv^T\allocv=\sum_{\type\in\Types}\prior(\type)\mathbb{E}_{\conditional}\left[\tilde{\beliefv}^T\linearv\right]=\mathbb{E}_{\unconditional}\left[\tilde{\beliefv}^T\linearv\right]=\priorv^T\linearv.
\end{align}
However, as the results in this section illustrate, different information structures lead to different welfare profiles, so that the chosen information structure determines how the different types in the population share the ex ante welfare. 

%


When \welfare\ is as in \autoref{eq:linear}, we can provide an alternative characterization of the set \repset. From \autoref{sec:main}, it follows that $\allocv\in\repset$ if and only if we can find a Bayes plausible distribution over posteriors, $\bsplit\in\bplausible$, with finite support, such that
 \begin{align}\label{eq:linear-concavification}\allocv=\expect_{\tau}\left[\Reputationv(\tilde{\beliefv})\right]=\expect_{\tau}\left[\frac{\tilde{\beliefv}}{\priorv}(\tilde{\beliefv}^T\linearv)\right]=\dprior\expect_{\tau}\left[\tilde{\beliefv}\tilde{\beliefv}^T\right]\linearv,
 \end{align}
 where $\dprior$ denotes a diagonal matrix with $(i,i)$-th element equal to $\nicefrac{1}{\prior(\type_i)}$.
 
 \autoref{eq:linear-concavification} shows a \bprofile\ can be represented as the product of three terms: the reputation vector \linearv, the prior-normalizing matrix \dprior, and the matrix $\expect_{\tau}[\beliefv\beliefv^T]$. Furthermore, the matrix $\expect_{\tau}\left[\beliefv\beliefv^T\right]$ satisfies the following two properties. First, it is a \emph{completely positive matrix} \citep{berm88}: an $N\times N$ matrix $\cpmat$ is completely positive if it can be written as $\sum_{m=1}^M\vectorsv_{\bs{m}}\vectorsvT_{\bs{m}}$ for some finite collection of non-negative vectors $\vectorsv_{\bs{m}}\in\mathbb{R}^N_+$.\footnote{Completely positive matrices have been studied extensively as they play an important role in optimization theory, machine learning, and other applications \citep{bersha2003}. A completely positive matrix is symmetric and positive-semidefinite, with positive elements; for $N\leq4$, the converse is also true.} Second, the rows of the matrix $\expect_{\tau}\left[\beliefv\beliefv^T\right]$ add up to the prior: $\expect_{\tau}\left[\beliefv\beliefv^T\right]\ones=\expect_{\tau}\left[\beliefv(\beliefv^T\ones)\right]=\expect_{\tau}\left[\beliefv\right]=\priorv$. \autoref{theorem:linear-concavification} shows these two properties are not only necessary but also sufficient and thus fully characterize the \bset:\footnote{\label{fnt:cp}An analogous characterization appears in concurrent work by \cite{saba21}, who discuss the computational advantages of working with completely positive matrices to solve information design problems.}
 
\begin{theorem}\label{theorem:linear-concavification}
Given the reputation vector \linearv, $\allocv\in\repset$ if and only if a completely positive matrix $\cpmat\in\mathbb{R}^{N\times N}$ exists such that  $\cpmat \ones=\priorv$ and 
\begin{align*}
\allocv=\dprior\cpmat\linearv.
\end{align*}
\end{theorem}
Putting together the properties in \autoref{theorem:linear-concavification}, we obtain that any \bprofile\ $\allocv$ is the product of the reputation vector \linearv\ and a matrix \transmat, where $\transmat\equiv \dprior\cpmat$ is the transition matrix of a reversible Markov chain with invariant distribution \priorv. That is, (i) $\priorv^T\transmat=\priorv^T$, (ii) $\transmat \ones=\ones$,  and (iii) $\transmat$ satisfies the \emph{detailed balance conditions}: for all $i,j\in N$, $\upmu_{0i}\transmat_{ij}=\upmu_{0j}\transmat_{ji}$. The first property captures the pure redistribution of welfare highlighted in \autoref{eq:linear-total-welfare}. The second property implies any \bprofile\ can be viewed as a \emph{garbled} version of the full information profile \linearv. The third property delineates the limits of how payoffs can be redistributed by linking how much of $\linear(\type_i)$ can be attributed to $\type_j$, and vice versa. Indeed, because \transmat\ is the transition matrix of a reversible Markov chain, we obtain that there is \emph{mean reversion} in the redistribution of payoffs across types. To see this, note that if $\allocv=\transmat\linearv\in\repset$, then also $\transmat\allocv\in\repset$.\footnote{ $\transmat^2\ones=\transmat \ones=\ones$ and $\transmat^2=\dprior\cpmat^\prime$, where $\cpmat^\prime\equiv \cpmat \dprior\cpmat$ is completely positive because $\cpmat$ is symmetric.} Because \priorv\ is the invariant distribution of \transmat, we have that $\transmat^k\allocv\rightarrow_{k\rightarrow\infty}(\priorv^T\allocv)*\ones=(\priorv^T\linearv)*\ones=\allocv^{ND}$, where $\allocvN$ is the no-disclosure profile.

\begin{remark}[Connections to the literature]
Reversible Markov chains are prominent in the study of higher-order beliefs and expectations of higher-order beliefs \citep{samet1998iterated,cripps2008common,golub2017higher}. Indeed, note that when the welfare function is linear and type independent, a \bprofile\ is a vector of second-order expectations: for any type \type\ and any $\allocv\in\repset$, $\welfare(\type)$ is the expectation under some information structure of the random variable $\beliefv^T\linearv$ of an \agent\ that knows \type. Whereas that literature takes the \informationpolicy\ as given and shows  (sequences of) higher-order expectations can be obtained by iteratively applying the transition matrix of a reversible Markov chain, \autoref{theorem:linear-concavification} identifies \emph{which} transition matrices are consistent with \emph{some} information structure and shows complete positivity is the key property they must satisfy.

\autoref{theorem:linear-concavification} also relates to the literature on majorization \citep{hardy1952inequalities}: if all types are equally likely,  \transmat\ is doubly stochastic and \linearv\ majorizes \allocv. However, not any profile majorized by \linearv\ is a \bprofile, because not all doubly stochastic matrices are symmetric, and hence, some do not satisfy the detailed balance conditions.
\end{remark}

We now show how \autoref{theorem:linear-concavification} delivers a more general version of the truth-drifting property discussed in \autoref{sec:main}. This property has been obtained in different forms in the literature that studies
 the feasible evolution of beliefs (e.g., \citealp{frkr14}, \citealp{hart2020posterior}). Truth-drifting states that whereas an \informationpolicy\ can occasionally ``deceive'' the outside observer about an \agent's true type, it cannot systematically do so. This property underlies the limits of using information as a tool to distribute welfare in the population.
%

Formally, consider any event $\event$ that is correlated with the types according to the conditional probability function $\eventcorr\in[0,1]^N$, $\eventcorr_i\equiv\Pr(\event\mid \type_i)$, so that the prior probability of the event is $\Pr(\event)=\priorv^T\eventcorr$.\footnote{For concreteness, \event\ can be seen as a subset of $\types\times[0,1]$ equipped with a probability measure that agrees with \prior\ on \types\ \citep{grst78,geka17}.} If all $\eventcorr_i\in\{0,1\}$, the event effectively indicates a subset of types. More generally, the event may involve extraneous uncertainty, and the types may be only imperfectly informative about it. We show that if the event is true, the average posterior probability that the outside observer attaches to this event must be at least as large as the prior probability:
\begin{claim}[Truth drifting]\label{claim:truth-drifting}
For any event $\event$ and \informationpolicy\ \policy,
\begin{align*}
    \expect_{\policy}\left[\Pr(\event\mid\signal)\mid\event\right]\geq \Pr(\event).
\end{align*}
\end{claim}
\cite{frkr14} obtain this result, relying on the properties of Kullback-Leibler divergence.\footnote{In a setting in which information respects the state space's ordinal structure, \cite{kartik2021information} formalize the sense in which the drift toward the truth is stronger for more Blackwell informative \informationpolicies.} \cite{hart2020posterior} obtain a version of 
\autoref{claim:truth-drifting} in the special case of $\event\subseteq\types$, relying on the monotone-likelihood ratio property. Instead, our proof of \autoref{claim:truth-drifting}, presented in \autoref{appendix:theorem-linear-concavification}, builds on the property that the underlying matrix \cpmat\ is completely positive and thus necessarily positive semi-definite.

\paragraph{Boundary \informationpolicies:} Recall that \autoref{theorem:hyperplanes} characterizes the boundary of the \bset\ by means of supporting Bayesian persuasion problems. In the reputation model, \autoref{eq:linear-total-welfare} implies the \bset\ lies within a hyperplane with orthogonal vector $\priorv$.  Thus, instead of studying the boundary of the \bset, we focus on its \emph{relative} boundary, which consists of all \bprofiles\ not in the relative interior of the \bset.\footnote{Recall that the relative interior of a set $X$ is the interior of $X$ within its affine hull, which is the set of all affine combinations of elements in $X$.} In a slight abuse of terminology, we refer to the profiles on the relative boundary of \repset\ and the \informationpolicies\ that induce them as \emph{boundary} profiles and \informationpolicies, respectively. 


As we show next, the supporting Bayesian persuasion problems in the reputation model take a well-known structure. Indeed, fix a direction $\directionv\in\reals^N\setminus\{0\}$ not collinear with $\priorv$ and consider the induced supporting Bayesian persuasion problem:
\begin{align}\label{eq:rayo-segal}\tag{RS$_\directionv$}
    \max_{\tau\in\bplausible}\expect_{\tau}\left[\directionv^T\Reputationv(\beliefv)\right]&=\max_{\tau\in\bplausible}\expect_\tau\left[\left(\frac{\directionv}{\priorv}^{T}\beliefv\right)\left(\linearv^T\beliefv\right)\right]
     \nonumber\\
  & 
 =\max_{\tau\in\bplausible}\expect_\tau\left[\expect_{\belief}\left[\frac{\direction(\type)}{\prior(\type)}\right]\expect_{\belief}\left[\linear(\type)\right]\right],\nonumber
\end{align}
where the first equality uses the form of $\reputation$ and the definition of $\Reputation$. \autoref{eq:rayo-segal} shows that if an \informationpolicy\ \policy\ delivers a profile \allocv\ on the relative boundary of \repset, the \informationpolicy\ solves an instance of the information design problem in \cite{rayo2010optimal}. To be precise, \cite{rayo2010optimal} consider the following problem. A sender owns a prospect, and his objective is that the receiver accepts it. When the sender's type is \type\ and the receiver accepts the prospect, the sender and the receiver obtain a payoff $\gamma(\type)\equiv\direction(\type)/\prior(\type)$ and $\linear(\type)\in[0,1]$, respectively. Instead, if the receiver rejects the prospect, the sender obtains a payoff of $0$, whereas the receiver obtains a payoff $u$ distributed uniformly over $[0,1]$ independently of $\type$. The sender chooses an \informationpolicy, \policy, without observing the realization of $u$. Thus, when \policy\ induces a belief \beliefv,  the sender expects  the receiver to accept the project with probability, $\linearv^T\beliefv$. It follows that the last term in \autoref{eq:rayo-segal} represents the sender's expected payoff when \bsplit\ is the distribution over posteriors induced by \informationpolicy\ \policy.

\begin{proposition}\label{proposition:linear-hyperplanes}
\emph{(Boundary profiles)} 
A welfare profile \allocv\ is on the relative boundary of \repset\ if and only if a direction $\directionv\in\reals^N\setminus\{0\}$ not collinear with \prior\ exists such that $\allocv$ is induced by an \informationpolicy\ that solves the program \ref{eq:rayo-segal}.
%
\end{proposition}
\autoref{proposition:linear-hyperplanes} allows us to rely on the approach of \cite{rayo2010optimal} to characterize the \emph{shape} of the \informationpolicies\ that achieve the boundary \bprofiles. This approach relies on a graphical representation of an \informationpolicy, in which the prospect values $\{(\gamma(\type),\linear(\type)):\type\in\types\}$ are the nodes (see \autoref{remark:rs-appendix} in the appendix). 
The results in \cite{rayo2010optimal} have immediate implications for the information structures that induce the boundary profiles of \repset:
\begin{corollary}\label{corollary:linear-pareto} In the reputation model, the following hold:
\begin{enumerate}
\item An \informationpolicy\ that induces a boundary \bprofile\ in the direction $\directionv$ does not pool types $\type_i$ and $\type_j$  whenever their ranking under the vector $\directionv/\priorv$ and the reputation vector  \linearv\ is the same;
\item The full- and no-disclosure profiles are on the relative boundary of \,\repset.
\end{enumerate}
\end{corollary}
The first part of \autoref{corollary:linear-pareto} highlights a natural feature of optimal information provision in the reputation model in terms of the alignment of preferences of the social planner, captured by $\directionv/\priorv$, and of the outside observer, captured by \linearv: the planner should not pool any two types as long as the planner and the outside observer are in agreement about the types' relative ranking. The second part of \autoref{corollary:linear-pareto} shows that the full- and no-disclosure profiles are on the boundary of the \bset. Whereas this property holds in Examples \ref{example:privacy} and \ref{example:platforms}, in which the welfare function is not linear, this property is not a general one, as we illustrate in \autoref{example:fd-nd-not-bd} in the appendix.


%
%
\paragraph{Individual reputation bounds:} We can further build on the graphical approach of \cite{rayo2010optimal} to characterize the \informationpolicies\ that deliver maximal (or minimal) welfare to any given type. This exercise provides a rough way to bound the \bset\ \repset\ and also suggests how information  may be employed to boost (or dilute) the reputation of particular types in the population. In the context of our credit agency example, \autoref{example:credit}, the \informationpolicy\ that maximizes the welfare of individuals of type \targettype\ maximizes the probability that \agents\ of type $\type_i$ obtain credit.

Formally, given a target type \targettype, we want to solve  the following problem:
\begin{align}\label{eq:max-i}\tag{$i$-MAX}
    \max_{\allocv\in\repset}\allocvi.
\end{align}
\autoref{proposition:bi-pooling} below shows a particular class of \informationpolicies\ solves the problem \ref{eq:max-i}.
%
\begin{definition}[Noisy priority]\label{definition:bipooling}
A \emph{\targettype-noisy-priority policy with threshold $k$} is an \informationpolicy\ $(\excdf,\signals)$ such that
$\signals=\types$, and the likelihood function \excdf\ satisfies:
\begin{enumerate}
    \item If $j\neq i$, $\pi(s=\type_j\mid\type_j)=1$,
    \item If $j<k$, $\pi(s=\type_j\mid\type_i)=0$, and
    \item If $j\geq k$, $\pi(s=\type_j\mid\type_i)>0$.
\end{enumerate}
\end{definition}
In other words, a \targettype-noisy-priority policy pairwise pools the target type $\targettype$ with all types with indices above some threshold and separates all other types. A \bipooling\ has an implementation akin to the priority mechanisms in the matching literature \citep{celebi2022adaptive}, and hence its name. A \bipooling\ with threshold $k\geq i$ can be implemented by first assigning a perfectly revealing score to each type equal to their index, and then prioritizing the target type \targettype\ by increasing this type's score by a random number. 
\begin{proposition}\label{proposition:bi-pooling} The \bprofile\ that solves \ref{eq:max-i} is induced by a \targettype-\bipooling\ with threshold $k\geq i$.
\end{proposition}
The proof is in \autoref{appendix:theorem-linear-concavification}. One part of \autoref{proposition:bi-pooling} is straightforward: if one wishes to increase the reputation of $\type_i$, then $\type_i$ should be separated from all types with lower indices. What might be less obvious is that whenever $\type_i$ is pooled with some other type, $\type_i$ should be pooled with it pairwise. In a sense, pooling several types together redistributes the reputation from higher-quality types to lower-quality types. Pairwise pooling then allows the target type to obtain maximal reputation gains from any other type without sharing the gains with others. Finally, pairwise pooling with many types ensures no signal is overly ``muddled,'' which in turn ensures an overall high reputation for \targettype.

By simply reversing signs, \autoref{proposition:bi-pooling} can be used to characterize the \informationpolicy\ that minimizes the expected reputation of \agents\ of type \targettype: this information structure should pairwise pool the target type with types whose indices are below some threshold. Such adversarial pairwise pooling inflicts maximal reputation losses and can be viewed as a \emph{noisy-degrading policy}.

We conclude this section by illustrating the \bset\ in the context of \autoref{example:credit} and highlight the importance of population heterogeneity as captured by the number of types.
\begin{example}[continued]\label{example:linear} \autoref{fig:linear-2} illustrates the results of this section in the context of \autoref{example:credit}. Recall that in this case $\linear(\type)$ denotes the probability that an \agent\ of type \type\ repays the loan, and hence, $\welfare(\belief,\type)$ is the expected repayment probability under belief \belief. Like in \cite{rayo2010optimal}, we assume the credit agency has a uniform outside option. Assuming \agents\ wish to maximize the probability the lending agency approves the loan justifies that $\linearv^T\beliefv$ corresponds to their welfare. 

\autoref{fig:linear-2-2} depicts the individually feasible welfare profiles (dashed square) and the \bset\ (blue line) in the case of $N=2$. \autoref{proposition:bounds} implies any payoff between $\linearv_1=0$ and $\priorv^T\linearv$ is feasible for $\type_1$, whereas any payoff between $\priorv^T\linearv$ and $\linearv_2=1$ is feasible for $\type_2$. As \autoref{fig:linear-2-2} illustrates, the Cartesian product $[\linearv_1,\priorv^T\linearv]\times[\priorv^T\linearv,\linearv_2]$  is a rather lax bound in this example. In particular, the Cartesian product $[\linearv_1,\priorv^T\linearv]\times[\priorv^T\linearv,\linearv_2]$ ignores that all \bprofiles\ satisfy $\priorv^T\allocv=\priorv^T\linearv=0.5$ (\autoref{eq:linear-total-welfare}). In the case of $N=2$, adding this restriction is enough to pin down the \bset. The reason is that by \autoref{theorem:linear-concavification}, all \bprofiles\ can be obtained by ``garbling'' the full-disclosure \bprofile, \linearv, and in the case of binary types, this garbling turns out to span a linear segment. Finally, note the structure of the \bset\ implies a social planner with Pareto weights \directionv\ finds it optimal to provide no or full information, depending on whether the planner weighs the welfare of $\type_1$-\agents\ more than that of $\type_2$-\agents\ (that is, $\directionv_1\lessgtr\directionv_2$).

\begin{figure}[t!]
        \centering
\subfloat[$N=2$, $\linearv=(0,1)$, $\priorv=(1/2,1/2)$]{\scalebox{0.9}{%

    \begin{tikzpicture}[scale=0.8,xscale=1,yscale=1.2]
    \begin{axis}[xtick pos=left,
    ytick pos=left,
            xmin=-0.25,xmax=0.75,
        ymin=0.25,ymax=1.25,
        xtick={0,0.5},
        ytick={0.5,1},
        xlabel=$\allocv_1$,ylabel=$\allocv_2$
         ,
         x label style={at={(axis description cs:0.5,-0.05)},anchor=north},
    y label style={at={(axis description cs:-0.05,.5)},rotate=-90,anchor=south}
]
\addplot[mark=none, color=blue!30, very thick] coordinates{(0,1) (0.5,0.5)};
\addplot[mark=none,color=black,dashed] coordinates{(0,1) (0.5,1)};
\addplot[mark=none,color=black,dashed] coordinates{(0.5,0.5) (0.5,1)};
\addplot[mark=none,color=black,dashed] coordinates{(0,0.5) (0.5,0.5)};
\addplot[mark=none,color=black,dashed] coordinates{(0,1) (0,0.5)};
\addplot[color=blue,mark=*]coordinates{(0,1)}
[every node/.style={yshift=8pt}]
node[pos=0] {$\color{black}{\allocv^{FD}}$};
\addplot[color=red,mark=*]coordinates{(0.5,0.5)}
[every node/.style={xshift=0pt,yshift=-10pt}]
node[pos=0] {$\color{black}{\allocv^{ND}}$};
    \end{axis}
    \end{tikzpicture}    
 }\label{fig:linear-2-2}}\hspace{0.25cm}
 \subfloat[$N=3$, $\linearv=(0,0.5,1)$, $\priorv=(1/3,1/3,1/3)$]{\scalebox{0.9}{
  \begin{tikzpicture}
    \begin{axis}[xtick pos=left, ytick pos= lower, ztick pos=left, xmin=-0.05,xmax=0.6,ymin=0.2,ymax=0.8,zmin=0.4,zmax=1.05,
   xtick={0.2,0.4},ytick={0.4,0.6}, ztick={0.6,0.8,1},
    xlabel={$\allocv_1$}, ylabel={$\allocv_2$}, zlabel={$\allocv_3$},z label style={rotate=-90},view/h=80, view/v=10
    ]
    
    \addplot3[mark=none,color=black,dotted] coordinates 
{
    (0,0.25,0.471)
    (0.529,0.25,0.471)
    (0.529,0.75,0.471)
    (0,0.75,0.471)
     (0,0.25,0.471)
};
    
    \addplot3[mark=none,color=black,dotted] coordinates 
{
        (0,0.75,0.471)
           (0,0.75,1)
};
    
\addplot3[surf, mesh/interior colormap=
           {blackblack}{color=(blue!30) color=(blue!30)},
        mesh/interior colormap thresh=1, 
    shader=flat]
            coordinates {(0.5, 0.5, 0.5) (0.504753, 0.490196, 0.505051) (0.509027, 0.480769,
  0.510204) (0.512838, 0.471698, 0.515464) (0.516204, 0.462963, 
  0.520833) (0.519139, 0.454545, 0.526316) (0.521657, 0.446429, 
  0.531915) (0.523769, 0.438596, 0.537634) (0.525487, 0.431034, 
  0.543478) (0.526821, 0.423729, 0.549451) (0.527778, 0.416667, 
  0.555556) (0.528366, 0.409836, 0.561798) (0.528592, 0.403226, 
  0.568182) (0.528462, 0.396825, 0.574713) (0.52798, 0.390625, 
  0.581395) (0.527149, 0.384615, 0.588235) (0.525974, 0.378788, 
  0.595238) (0.524456, 0.373134, 0.60241) (0.522597, 0.367647, 
  0.609756) (0.520397, 0.362319, 0.617284) (0.517857, 0.357143, 
  0.625) (0.514976, 0.352113, 0.632911) (0.511752, 0.347222, 
  0.641026) (0.508184, 0.342466, 0.649351) (0.504267, 0.337838, 
  0.657895) (0.5, 0.333333, 0.666667) (0.495377, 0.328947, 
  0.675676) (0.490393, 0.324675, 0.684932) (0.485043, 0.320513, 
  0.694444) (0.479319, 0.316456, 0.704225) (0.473214, 0.3125, 
  0.714286) (0.46672, 0.308642, 0.724638) (0.459828, 0.304878, 
  0.735294) (0.452527, 0.301205, 0.746269) (0.444805, 0.297619, 
  0.757576) (0.436652, 0.294118, 0.769231) (0.428052, 0.290698, 
  0.78125) (0.418993, 0.287356, 0.793651) (0.409457, 0.284091, 
  0.806452) (0.399429, 0.280899, 0.819672) (0.388889, 0.277778, 
  0.833333) (0.377817, 0.274725, 0.847458) (0.366192, 0.271739, 
  0.862069) (0.35399, 0.268817, 0.877193) (0.341185, 0.265957, 
  0.892857) (0.327751, 0.263158, 0.909091) (0.313657, 0.260417, 
  0.925926) (0.298872, 0.257732, 0.943396) (0.283359, 0.255102, 
  0.961538) (0.267083, 0.252525, 0.980392) (0.25, 0.25, 1.) 
  (0, 0.5, 1)
  
(0.5, 0.5, 0.5) (0.494949, 0.509804, 0.495247) (0.489796, 0.519231,
  0.490973) (0.484536, 0.528302, 0.487162) (0.479167, 0.537037, 
  0.483796) (0.473684, 0.545455, 0.480861) (0.468085, 0.553571, 
  0.478343) (0.462366, 0.561404, 0.476231) (0.456522, 0.568966, 
  0.474513) (0.450549, 0.576271, 0.473179) (0.444444, 0.583333, 
  0.472222) (0.438202, 0.590164, 0.471634) (0.431818, 0.596774, 
  0.471408) (0.425287, 0.603175, 0.471538) (0.418605, 0.609375, 
  0.47202) (0.411765, 0.615385, 0.472851) (0.404762, 0.621212, 
  0.474026) (0.39759, 0.626866, 0.475544) (0.390244, 0.632353, 
  0.477403) (0.382716, 0.637681, 0.479603) (0.375, 0.642857, 
  0.482143) (0.367089, 0.647887, 0.485024) (0.358974, 0.652778, 
  0.488248) (0.350649, 0.657534, 0.491816) (0.342105, 0.662162, 
  0.495733) (0.333333, 0.666667, 0.5) (0.324324, 0.671053, 
  0.504623) (0.315068, 0.675325, 0.509607) (0.305556, 0.679487, 
  0.514957) (0.295775, 0.683544, 0.520681) (0.285714, 0.6875, 
  0.526786) (0.275362, 0.691358, 0.53328) (0.264706, 0.695122, 
  0.540172) (0.253731, 0.698795, 0.547473) (0.242424, 0.702381, 
  0.555195) (0.230769, 0.705882, 0.563348) (0.21875, 0.709302, 
  0.571948) (0.206349, 0.712644, 0.581007) (0.193548, 0.715909, 
  0.590543) (0.180328, 0.719101, 0.600571) (0.166667, 0.722222, 
  0.611111) (0.152542, 0.725275, 0.622183) (0.137931, 0.728261, 
  0.633808) (0.122807, 0.731183, 0.64601) (0.107143, 0.734043, 
  0.658815) (0.0909091, 0.736842, 0.672249) (0.0740741, 0.739583, 
  0.686343) (0.0566038, 0.742268, 0.701128) (0.0384615, 0.744898, 
  0.716641) (0.0196078, 0.747475, 0.732917) (0., 0.75, 0.75) 
  (0, 0.5, 1)
  };
\addplot3[mark=none,color=black,dotted] coordinates 
{
    (0,0.25,1)
    (0.529,0.25,1)
    (0.529,0.75,1)
    (0,0.75,1)
     (0,0.25,1)
};
\addplot3[mark=none,color=black,dotted] coordinates 
{
  (0,0.25,0.471)
     (0,0.25,1)
};
\addplot3[mark=none,color=black,dotted] coordinates 
{
     (0.529,0.25,0.471)
        (0.529,0.25,1)
};
\addplot3[mark=none,color=black,dotted] coordinates 
{
        (0.529,0.75,0.471)
           (0.529,0.75,1)
};

    \addplot3[color=blue,mark=*]coordinates{(0,0.5,1)}
[every node/.style={yshift=4pt}]
node[pos=0] {$\color{black}{\allocv^{FD}}$};
\addplot3[color=red,mark=*]coordinates{(0.5,0.5,0.5)}
[every node/.style={xshift=10pt,yshift=-15pt}]
node[pos=0] {$\color{black}{\allocv^{ND}}$};
    
    \end{axis}

    \end{tikzpicture}
    }\label{fig:linear-3}}
 
        \caption{Expected reputation in \autoref{example:credit}. The blue color marks   \bset\ \repset. The dashed segments outline the Cartesian product of individual welfare sets.}\label{fig:linear-2}
            \end{figure}

\autoref{fig:linear-3} depicts the \bset\ \repset\ in the case of $N=3$. In contrast to the binary-type case, the boundary of the \bset\ is non-linear and features a continuum of extreme points. This example illustrates that the constraints imposed by Bayes plausibility are richer than the simple garbling constraint that characterizes the set when $N=2$. As we show in \autoref{appendix:theorem-linear-concavification}, four classes of \informationpolicies\ span the boundary of the \bset. The first two are the $\type_1$-noisy-priority and $\type_3$-noisy-degrading policies, which span the nonlinear segment of the boundary. The other two span the linear segments in \autoref{fig:linear-3} and similar to the $\type_3$-noisy-priority and the $\type_1$-noisy-degrading policies, either maximize the loan-approval probability of $\type_3$-individuals, by separating them from $\type_1$- and $\type_2$-\agents, or minimize the loan-approval probability of $\type_1$-individuals, by separating them from $\type_2$- and $\type_3$-\agents. Unlike the $\type_3$-noisy-priority and the $\type_1$-noisy-degrading policies, the two classes of information structures that span the linear segments may pool together the types below $\type_3$ or above $\type_1$, respectively. Notably, for almost all boundary points, $\type_2$-\agents\ are pooled with \agents\ of some other type. Intuitively, $\type_2$-\agents\ exert \emph{pooling externalities} on \agents\ of types $\type_1$ and $\type_3$.\footnote{We follow the terminology in \cite{galperti2022value}, who highlight that \agents\ with certain types are valuable precisely because of the possibility of pooling them with other types.} For instance, $\type_2$-\agents\ enable boosting the loan-approval probability of $\type_1$-\agents\ in the case of the $\type_1$-\bipooling, but exert a negative externality on $\type_3$-\agents\ in the $\type_3$-noisy-degrading policy.\hfill$\closing$
\end{example}

\section{Conclusion}\label{sec:conclusion}\color{black}
We provide a framework to study the potentially disparate impact of information policies in a population of heterogeneous \agents. Because information policies increasingly shape society's choices in high-stakes domains, the \bset\ describes the limits of what society can achieve under such policies and the welfare trade-offs implied by the choice between different information policies. In the spirit of mechanism design and information design, our characterization of the \bset\ provides a unifying tool to evaluate the welfare implications of different information policies across a wide array of objective functions. 

We see several avenues worth exploring and left for future work. First, our model assumes  any information can be provided about an \agent's payoff-relevant \characteristic. However, this assumption does not necessarily hold in applications of interest in which an \agent's type may include protected characteristics. In the online appendix, we extend our framework to accommodate limits on how much information can be disclosed about the \agents\ in the population. The analysis there, however, does not consider that these limits may be designed when a potentially malicious third party selects the \informationpolicy. \cite{liang2022algorithmic} consider this case in the context of  decision-making algorithms, and we expect their insights to extend to the case of recommendations algorithms like the ones we consider. 

Second, because the welfare function depends only on the first-order beliefs about an \agent's type, our model only accounts for strategic interactions that follow the realization of a public signal (see, e.g., \citealp{laclau2017public}). Extending the analysis to account for general strategic interactions is worth exploring. \cite{galperti2022value}, which studies the \bprofile\ that gives rise to the sender's maximum average payoff across all Bayes correlated equilibria, is a step in this direction.


Third, motivated by recent policies, \cite{tirl20} studies the use of information in the form of a social score to incentivize good behavior in the population. Whereas \cite{tirl20} studies this question in the context of a parametric family of information structures, our initial explorations show the \bset\ allows us to extend his results by allowing \emph{any} information structure. More generally, information has been suggested as a substitute for monetary incentives, and the \bset\ describes what can be achieved with information alone. 

Finally, whereas we characterize \agents' welfare as a function of their type,  thinking of applications in which we care instead about the welfare of \emph{groups} is natural. For instance, an \agent's type could encompass their gender and their ability, and the social planner is concerned with the welfare different genders may obtain. 
This extension can inform the study of statistical discrimination, where recent work shows Bayesian persuasion tools can shed new light to this problem \citep{chambers2021characterisation,escude2022statistical,deb2022which}. Whereas much of the existing literature focuses on statistical properties of discrimination, our work highlights the important aspect of economic welfare, offering a complementary perspective that enriches the ongoing dialogue between statistics and economics.
\color{black}

\bibliographystyle{ecta}

\bibliography{general}
\appendix
\section{Omitted results and proofs from the main text}
\subsection{Omitted proofs from \autoref{sec:main}}
For completeness, we include a proof of \autoref{claim:accounting}, which follows from the analysis in \citet[pp.~683--684]{alonso2016bayesian}, in present notation:
\begin{proof}[Proof of \autoref{claim:accounting}]
For an \informationpolicy\ \policy,  let $\mathrm{supp}\unconditional$ denote the support of the distribution over posterior beliefs induced by \policy\ and let $\mathrm{Pr}_\policy(s)$ denote the unconditional probability of signal $s$ under the prior distribution $\prior$, i.e., $\mathrm{Pr}_\policy(s)=\sum_{\type\in\Types}\prior(\type)\excdf(s|\type)$. Then, for a given type \type, their welfare under \informationpolicy\ \policy\ can be written as follows:
\begin{align}\label{eq:accounting-app}
\welfare_\policy(\type)=\expect_{\conditional}\left[\reputation(\beliefc, \type)\right]&=\sum_{\belief\in \mathrm{supp}\unconditional}\sum_{\signal\in S:\belief_\signal=\belief}\pi(\signal|\type)\reputation(\belief,\type)\\
&=\sum_{\belief\in\mathrm{supp}\unconditional}\sum_{\signal\in S:\belief_\signal=\belief}\Pr\nolimits_{\policy}(\signal)\frac{1}{\prior(\type)}\frac{\prior(\type)\pi(\signal|\type)}{\Pr\nolimits_{\policy}(\signal)}\reputation(\belief,\type)\nonumber \\ 
&=\sum_{\belief\in\mathrm{supp}\unconditional}\sum_{\signal\in S:\belief_\signal=\belief}\Pr\nolimits_{\policy}(\signal)\frac{\belief(\type)}{\prior(\type)}\reputation(\belief,\type)\nonumber\\
&=\sum_{\belief\in\mathrm{supp}\unconditional}\sum_{\signal\in S:\belief_\signal=\belief}\Pr\nolimits_{\policy}(\signal)\Reputation(\belief,\type)=\expect_{\unconditional}\left[\Reputation(\beliefc,\type)\right].\nonumber
\end{align}
\end{proof}
\begin{proof}[Proof of \autoref{theorem:concavification}]
By definition, the point $\left(\priorv,\allocv\right)\in\mathrm{co}\left(\mathrm{graph\,}\Reputationv\right)$ if and only if a Bayes plausible distribution over posteriors \bsplit\ exists such that $\expect_\bsplit[\Reputationv\left(\belief\right)]=\allocv$. At the same time, the distribution over posteriors  induced by an \informationpolicy, \unconditional, is Bayes plausible, i.e., $\expect_{\unconditional}[\tilde{\belief}]=\prior$.
The result follows from \autoref{claim:accounting}.
\end{proof}
\paragraph{On closedness} \autoref{proposition:closed-bset} collects the results discussed in \autoref{sec:main}. To introduce \autoref{proposition:closed-bset}, we first introduce the analogue of the \bset\ when the population's welfare depends on the outside observer's actions and together with the \informationpolicy, we can flexibly choose a selection from the outside observer's best-response correspondence. Formally, upon observing the realization from the \informationpolicy, the outside observer takes an action in a finite set \actions\ to maximize their expected payoff. We denote the outside observer's utility by $u:\actions\times\types\mapsto\reals$ and the welfare function by $v:\actions\times\types\mapsto\reals$. Given an \informationpolicy, \policy, let $\mathrm{supp}\unconditional$ denote the support of the distribution over posteriors \unconditional. For each $\belief_s\in\mathrm{supp}\unconditional$, let 
\begin{align}\label{eq:best-response}
\selection(\belief_s)\in\Delta\left(\arg\max_{\aaction\in\actions}\sum_{\type\in\Types}\belief_s(\type)u(\aaction,\type)\right)\equiv\Delta\left(\aaction^*(\belief)\right),
\end{align}
denote the outside observer's (possibly mixed) best response. Recall that the Theorem of the Maximum implies that the correspondence $\aaction^*(\belief)$ is upper-hemicontinuous.

Denote  by $F$ set of tuples $(\policy,\selection)$, such that selection \selection\ satisfies \autoref{eq:best-response}. Each $(\policy,\selection)\in F$ defines a welfare profile, $\welfare_{\policy,\selection}:\Types\mapsto\reals^N$, such that
\begin{align*}
    \welfare_{\policy,\selection}(\type)=\sum_{s\in S}\excdf(s|\type)\sum_{\aaction\in\actions}\selection(\belief_s)(\aaction)v(\aaction,\type)=\mathbb{E}_{\conditional}\left[\sum_{\aaction\in\actions}\selection(\beliefc)(\aaction)v(\aaction,\type)\right].
\end{align*}

The analogue of the \bset, which we denote by \repsetbp, is then
\begin{align*}
    \repsetbp=\left\{\welfarev\in\reals^N:\left(\exists (\policy,\selection)\in F\right)\text{ s.t. }\welfarev_i=\welfare_{\policy,\selection}(\type_i)\; \forall i\in\{1,\dots,N\}\right\}.
\end{align*}

We have the following result:
\begin{proposition}\label{proposition:closed-bset}
    The following hold:
    \begin{enumerate}[label=(\alph*)]
        \item\label{itm:continuity} If for each $\type\in\Types$, the welfare function $\welfare(\cdot,\type)$ is continuous on \Posteriors, the \bset\ \repset\ is closed.
        \item\label{itm:bp} The set \repsetbp\ is closed.
    \end{enumerate}
\end{proposition}

    

\begin{proof}[Proof of \autoref{proposition:closed-bset}]
The proof of part \ref{itm:continuity} is immediate and hence omitted. Consider then part \ref{itm:bp}. Define the analogue of the truth-adjusted welfare function $\hat{v}(\aaction,\belief,\type)$ to be
\begin{align*}
    \hat{v}(\aaction,\belief,\type)=\frac{\belief(\type)}{\prior(\type)}v(\aaction,\type).
\end{align*}
The same arguments as in \autoref{claim:accounting} imply that for any pair $(\policy,\selection)\in F$ and all $\type\in\Types$,
\begin{align*}
    \welfare_{\policy,\selection}(\type)=\mathbb{E}_{\unconditional}\left[\sum_{\aaction\in\actions}\selection(\beliefc)(\type)\hat{v}(\aaction,\beliefc,\type)\right].
\end{align*}
Define the correspondence $\vcorr:\Posteriors\rightrightarrows\reals^N$ as follows
\begin{align}
    \vcorr(\belief)=\mathrm{co}\left\{\mathrm{\hat{v}}(\aaction,\belief):\aaction\in\aaction^*(\belief)\right\}.
\end{align}
Observe that $\vcorr$ is a non-empty valued, convex-valued, and compact-valued correspondence. Furthermore, $\vcorr$ is upper-hemicontinuous. The first part follows immediately from noting that $\vcorr$ is the convex hull of finitely many vectors in $\reals^N$--$\aaction^*(\belief)$ is non-empty--and $\hat{v}$ is bounded. Upper-hemicontinuity of $\vcorr$ follows from continuity of \vrep\ in $\belief$ and upper-hemicontinuity of $\aaction^*(\belief)$. Consequently, the graph of $\vcorr$,
\begin{align*}
    \mathrm{graph}\; \vcorr=\left\{(\beliefv,\mathrm{v})\in\Posteriors\times\reals^N:\mathrm{v}\in \vcorr(\belief)\right\},
\end{align*}
is closed.

We now show that similar to \autoref{theorem:concavification}, the set \repsetbp\ is the section at the prior of the convex hull of the graph of the correspondence $\vcorr$, that is
    \begin{align}\label{eq:bset-bp-representation}
        \repsetbp=\left\{\welfarev\in\reals^N:(\priorv,\welfarev)\in\mathrm{co}\left(\mathrm{graph}\;\vcorr\right)\right\}.
    \end{align}
    Clearly, if $\welfarev\in\repsetbp$, then $(\priorv,\welfarev)\in\mathrm{co}(\mathrm{graph}\;\vcorr)$. To see that the opposite holds, let $(\priorv,\welfarev)\in\mathrm{co}(\mathrm{graph}\;\vcorr)$. Then, a finite collection $(\weights_k,\beliefv_k,\mathrm{\tilde{v}}_k)_{k=1}^M$ of non-negative weights, beliefs, and correspondence values exists such that $\sum_{k=1}^M\weights_k=1$, $\sum_{k=1}^M\weights_k\beliefv_k=\priorv$,  $\mathrm{\tilde{v}}_k\in \vcorr(\belief_k)$ for all $k\in\{1,\dots,M\}$, and 
    \begin{align}\label{eq:one}
        \welfarev=\sum_{k=1}^M\weights_k \mathrm{\tilde{v}}_k.
    \end{align}
    Because for each $k\in\{1,\dots,M\}$, $\tilde{\mathrm{v}}_k\in \vcorr(\belief_k)$, the definition of $\vcorr$ implies a finite collection of non-negative weights and actions, $\left\{\alpha_{k,l},\aaction_l\right\}_{l=1}^{L_k}$, exists such that $\sum_{l=1}^{L_k}\alpha_{l,k}=1$, for all $l\in\{1,\dots,L_k\}$,  $\aaction_l\in\aaction^*(\belief_k)$ and for all $i\in\{1,\dots,N\}$,
    \begin{align}\label{eq:two}
        \tilde{\mathrm{v}}_{k,i}=\sum_{l=1}^{L_k}\alpha_{k,l}\hat{v}(\aaction_l,\belief_k,\type_i).
    \end{align}
 Define $\policy$ to be the \informationpolicy\ with signals $S=\{\beliefv_1,\dots,\beliefv_M\}$, and signal distribution $\excdf(\belief_k|\type)=(\nicefrac{\belief_k(\type)}{\prior(\type)})\weights_k$. Furthermore,  define $\selection$ so that for belief $\belief_k$, $\selection(\belief_k)\in\Delta(\aaction^*(\belief_k))$ coincides with $\{\alpha_{k,l}\}_{l=1}^{L_k}$. By construction, $(\policy,\selection)\in F$. Equations \ref{eq:one} and \ref{eq:two} together imply that
    \begin{align}
        \welfarev=\sum_{k=1}^M\weights_k\sum_{l=1}^{L_k}\alpha_{k,l}\mathrm{\hat{v}}(\aaction_l,\belief_k)=\mathbb{E}_{\unconditional}\left[\sum_{\aaction\in\actions}\alpha(\beliefc)(\aaction)\mathrm{\hat{v}}(\aaction,\belief)\right].
    \end{align}
Finally, \autoref{eq:bset-bp-representation} allows us to conclude that the set \repsetbp\ is closed: it is the section at the prior of the convex hull of the graph of the correspondence $\vcorr$, which is closed.
\end{proof}

\subsection{Omitted proofs from \autoref{sec:pareto}}\label{appendix:pareto}
\begin{proof}[Proof of \autoref{theorem:hyperplanes}]
To complete the proof of the first direction of \autoref{theorem:hyperplanes}, we provide the steps to show that if $\welfarev\in\Pareto$, a direction $\directionv\in\reals_+^N\setminus\{0\}$ exists such that
    \begin{align*}
        \directionv^T\welfarev=\max\{\directionv^T\allocbv:\allocbv\in\repset\},
    \end{align*}
    The opposite direction in \autoref{theorem:hyperplanes} immediately follows. 

    Fix a Pareto efficient \welfarev\ and let $\Gamma=\{\allocbv\in\reals^N:\allocbv\geq\welfarev\}$. Clearly, $\Gamma$ is convex and $\mathrm{int}\;\Gamma$ is non-empty. Because \welfarev\ is Pareto efficient, $\mathrm{int}\;\Gamma\cap\repset=\emptyset$. By Minkowski's separating hyperplane theorem, a direction $\directionv\in\reals^N\setminus\{0\}$ exists such that for all $\welfarev^{\prime\prime}\in\repset$ and $\allocbv\in\Gamma$,
\begin{align*}
\directionv^T\welfarev^{\prime\prime}\leq\directionv^T\allocbv.
\end{align*}
Because $\welfarev\in\Gamma$, we have that $\directionv^T\welfarev^{\prime\prime}\leq\directionv^T\welfarev$ for all $\welfarev^{\prime\prime}$ in \bset. Thus, $\directionv^T\welfarev=\max\{\directionv^T\welfarev^{\prime\prime}:\welfarev^{\prime\prime}\in\repset\}$ (cf. \autoref{eq:hyperplanes}). Similarly, because $\welfarev\in\repset$, then we have that $\directionv^T\welfarev\leq\directionv^T\allocbv$ for all $\allocbv\in\Gamma$.

We now show that $\directionv\geq0$. Let \canonicali\ denote the canonical vector that has a 1 in coordinate $\type_i$ and $0$ otherwise. Then, $\welfarev+\canonicali\in\Gamma$, so that 
\begin{align*}
\directionv^T\welfarev\leq\directionv^T(\welfarev+\canonicali)\Rightarrow 0\leq\direction(\type_i).
\end{align*}
By definition, $\directionv\neq\mathrm{0}$ so that without loss of generality $\directionv\in\Posteriors$. 
\end{proof}
 \begin{proposition}\label{proposition:usc-general} 
Suppose $\welfare(\cdot,\type)$ is upper-semicontinuous for all $\type\in\Types$.\footnote{Taking \types\ to be a compact Polish space, we endow the set of Borel probability measures on \types, \Posteriors, and on \Posteriors, $\Delta(\Posteriors)$, with the weak$^*$ topology, so they are also compact Polish \citep[Theorems 15.11 and 15.12]{aliprantis2013infinite}.} Suppose $\welfarev^*$ is a limit point of \Pareto. Then, a point $\welfarev^{**}\in\Pareto$ exists such that $\welfarev^{**}\geq\welfarev^*$. Consequently, any weakly monotone social welfare function attains a solution in \repset. 
\end{proposition}
\begin{proof}[Proof of \autoref{proposition:usc-general}]
Let $(\welfarev_n)_{n\in\mathbb{N}}\subset\Pareto$ be such that $\welfarev_n\rightarrow\welfarev^*$ as $n\rightarrow\infty$. For each $n\in\mathbb{N}$ a direction $\direction_n\in\Posteriors$ exists such that 
    \begin{align}\label{eq:n-max}
        (\forall\allocbv\in\repset)\directionv_n^T\welfarev_n\geq\directionv_n^T\allocbv.
    \end{align}
    Since \Posteriors\ is compact, then up to a subsequence $\directionv_n\rightarrow\directionv_*$. Linearity of $\directionv_n^T\allocbv$ implies that taking limits on both sides of \autoref{eq:n-max} we obtain
    \begin{align}\label{eq:limit-max}
        (\forall\allocbv\in\repset)\directionv_*^T\welfarev^*\geq\directionv_*^T\allocbv.
    \end{align}
    Thus, if $\welfarev^*\in\repset$, \autoref{theorem:hyperplanes} implies that $\welfarev^*\in\Pareto$ and $\directionv_*\in\Posteriors$ is the direction that witnesses this. 
    
    Now, for each $\welfarev_n$, a distribution over posteriors $\bsplit_n\in\bplausible$ exists such that $\welfarev_n=\mathbb{E}_{\bsplit_n}\left[\Reputationv\right]$. Because \bplausible\ is compact, we have that, up to a subsequence, $\bsplit_n\rightarrow\bsplit^*\in\bplausible$. \citet[Theorem 15.5]{aliprantis2013infinite} implies $\mathbb{E}_\bsplit[\Reputationv]$ is upper-semicontinuous as a function of \bsplit, thus for all $i\in\{1,\dots,N\}$ we have that
    \begin{align}\label{eq:strong-pareto}
        \mathbb{E}_{\bsplit^*}[\Reputationv_i]\geq\lim_{n\rightarrow\infty}\mathbb{E}_{\bsplit_n}\left[\Reputationv_i\right]=\lim_{n\rightarrow\infty}\welfarev_{n,i}=\welfarev_i^*.
    \end{align}

    Let $\welfarev^{**}\equiv\mathbb{E}_{\bsplit^*}[\Reputationv]$ and note that it is an element of $\repset$ that dominates $\welfarev^*$ coordinate-by-coordinate. Moreover, because $\directionv_*\in\Posteriors$, \autoref{eq:strong-pareto} implies that $\directionv_*^T\welfarev^{**}\geq\directionv_*^T\welfarev^*$. 
    This, together with \autoref{eq:limit-max}, implies $\welfarev^{**}\in\Pareto$, which completes the proof.
\end{proof}

The proof of the statements in \autoref{corollary:benefits-new} follows from the following result:

\begin{corollary}[(No) Benefit from disclosure]\label{corollary:no-benefits} The following hold:
\begin{enumerate}
    \item\label{itm:yes-benefit}
    All types benefit from disclosure if for all $\type\in\Types$, $\Reputation(\cdot,\type)$ is strictly convex in a neighborhood of the prior. Furthermore, if  $\Reputation(\cdot,\type)$ is everywhere strictly convex, then full disclosure is uniquely Pareto efficient.
    \item\label{itm:no-benefit} No disclosure is Pareto efficient if either
    \begin{enumerate}
        \item\label{itm:single} a type $\type\in\Types$ exists such that $\Reputation(\belief,\type)$ is concave in \belief, or
        \item\label{itm:indep-cv} a vector $a\in\mathbb{R}_+^{N}$ and a concave function $\welfare:\Posteriors\mapsto\reals$ exist such that for all $\type\in\types$, the welfare function is given by $a(\type)\welfare(\belief)+b(\type)$.
    \end{enumerate}
\end{enumerate}
\end{corollary}

\begin{proof}[Proof of \autoref{corollary:no-benefits}]
Consider first the conditions in part \ref{itm:yes-benefit}. Toward a contradiction, suppose that \allocvN\ is in the Pareto frontier. Then, by \autoref{proposition:pareto-improvement} a direction $\tilde{\direction}\in\Posteriors$ exists such that no disclosure is a solution to the supporting Bayesian persuasion problem in direction $\tilde{\direction}$ . However, under the conditions in part \ref{itm:yes-benefit}, the indirect utility function is strictly convex in a neighborhood of the prior for all directions $\direction\in\Posteriors$, contradicting that no disclosure is a solution to the supporting Bayesian persuasion problem in some direction $\tilde{\direction}$ in \Posteriors. Furthermore, when $\Reputation(\cdot,\type)$ is strictly convex, so is $\hat{v}_\direction$ for all $\direction\in\Posteriors$ and the value of full disclosure strictly dominates that of any other \informationpolicy.

The proof of part \ref{itm:no-benefit} follows from \autoref{theorem:hyperplanes} by looking at the solution of the supporting Bayesian persuasion problem in certain directions. Under the conditions in part \ref{itm:single}, no disclosure is a solution to the supporting Bayesian persuasion problem in direction $\direction\in\Posteriors$ such that $\direction(\type)=1$ and $\direction(\typeb)=0$ for $\typeb\neq\type$. Instead, under the conditions in part \ref{itm:indep-cv}, no disclosure is a solution to the supporting Bayesian persuasion problem in direction $\direction(\type)=\prior(\type)/a(\type)$.
\end{proof}


\begin{proof}[Proof of \autoref{corollary:benefits-binary}]The result follows from \autoref{proposition:pareto-improvement}. For the first part, the condition ensures $\belief(\type)\welfare(\belief,\type)$ is strictly convex for all $\type\in\{\type_1,\type_2\}$, so the indirect utility function $\hat{v}_\direction$ is strictly convex for any $\direction\in\Posteriors$ and full disclosure is the unique solution to the supporting Bayesian persuasion problem. For the second part, the condition ensures $\belief(\type)\welfare(\belief,\type)$ is weakly concave for type $\type$, and the result follows from looking at the direction \directionv\ such that $\direction(\type)=1$ and $\direction(\typeb)=0$ for $\typeb\neq\type$.
\end{proof}

\subsection{Omitted proofs in \autoref{sec:linear}}\label{appendix:theorem-linear-concavification}
\begin{proof}[Proof of \autoref{theorem:linear-concavification}]
As explained in the main text, necessity follows from noting
\begin{align*}
    \allocv\in\repset\Rightarrow \allocv=\dprior\sum_{m=1}^M\alpha_m\beliefvm\beliefvm^T\linearv\equiv\dprior\cpmat\linearv,
\end{align*}
where $M\leq 2N$ follows from \autoref{corollary:cardinality}. \cpmat\ is completely positive because it is the convex combination of rank-one non-negative matrices, $\beliefvm\beliefvmT$. That $\cpmat \ones=\priorv$ follows from the martingale property of beliefs.

For sufficiency, consider $\allocv=\dprior\cpmat\linearv$, for some completely positive matrix $\cpmat$, such that $\cpmat \ones=\priorv$. Then,  $\{\vectorsvone,\dots,\vectorsvM\}\subseteq\mathbb{R}_+^N$ exist such that 
\begin{align}
    \cpmat=\sum_{m=1}^M\vectorsvm\vectorsvm^T.
\end{align}
Let $\sqrt{\alpha_m}=\sum_{j=1}^N\vectorsvm_j$ and note $\vectorsvm/(\sqrt{\alpha_m})\equiv\beliefvm\in\Posteriors$.
\begin{align*}
    \cpmat=\sum_{m=1}^M\alpha_m\left(\frac{\vectorsvm}{\sqrt{\alpha_m}}\right)\left(\frac{\vectorsvm}{\sqrt{\alpha_m}}\right)^T=\sum_{m=1}^M\alpha_m\beliefvm\beliefvm^T.
\end{align*}
It remains to show $\sum_{m=1}^M\alpha_m=1$ and $\sum_{m=1}^M\alpha_m\beliefvm=\priorv$. Note that for all $i\in\{1,\dots,N\}$,
\begin{align}
    (\cpmat \ones)_i=\sum_{m=1}^M\alpha_m\beliefv_{mi}\sum_{j=1}^N\beliefv_{mj}=\sum_{m=1}^M\alpha_m\beliefv_{mi}=\prior(\type_i).
\end{align}
Furthermore,
\begin{align}
    \sum_{i=1}^N\mu_{0}(\type_i)=1=\sum_{i=1}^N\sum_{m=1}^M\alpha_m\beliefv_{mi}=\sum_{m=1}^M\alpha_m.
\end{align}
Thus, an \informationpolicy\ exists that generates the distribution over posteriors $\{\alpha_m,\beliefvm\}_{m=1}^M$. Therefore, $\allocv\in\repset$.
\end{proof}
\begin{proof}[Proof of \autoref{claim:truth-drifting}]
If $\Pr(\event)=0$, the statement is trivial. If $\Pr(\event)>0$, denote by $\transmati$ the i-th row of the matrix $\transmat\equiv \dprior\cpmat$, presented as a row-vector. By Bayes' rule, $\Pr\left(\event\right)=\priorv^T\eventcorr$ and $\Pr(\type_i\mid \event)=(\upmu_{0i}\eventcorr_i)/(\priorv^T\eventcorr)$,
so
\begin{gather*}
\expect_{\policy}\left[\Pr(\event\mid \signal)\mid\type_i\right]=\sum_{j=1}^{N}\expect_{\policy}\left[\Pr\left[\type_j\mid \signal\right]\mid\type_i\right]\Pr(\event\mid\type_j)=\transmati\eventcorr,\\
\expect_{\policy}\left[\Pr\left(\event\mid \signal\right)\mid \event\right]=\sum_{i=1}^{N}\Pr\left(\type_i\mid \event\right)\expect_{\policy}\left[\Pr\left(\event\mid \signal\right)\mid\type_i\right]=\sum_{i=1}^{N}\frac{\upmu_{0i}\eventcorr_i}{\priorv^T\eventcorr}\transmati\eventcorr.
\end{gather*}
Hence, the truth-drifting condition can be restated as: 
\[
\sum_{i=1}^{N}\frac{\upmu_{0i}\eventcorr_i}{\priorv^T\eventcorr}\transmati\eventcorr\geq\priorv^T\eventcorr.
\]
Define $\cpmathat\equiv\transmat \dprior=\dprior\cpmat \dprior$. By \autoref{theorem:linear-concavification}, $\cpmathat$ is a completely positive matrix such that $\cpmathat\priorv=\ones$ and $\priorv^T\cpmathat\priorv=1$. Hence, the truth-drifting condition can be restated in a matrix form as: 
\[
\left(\frac{\priorv*\eventcorr}{\priorv^T\eventcorr}\right)^{T}\cpmathat\left(\frac{\priorv*\eventcorr}{\priorv^T\eventcorr}\right)\geq\priorv^T\cpmathat\priorv.
\]
The term $\zetav\equiv (\priorv*\eventcorr)/(\priorv^T\eventcorr)$ is an element of the simplex $\Delta(\types)$, equal to \priorv\ when $\eventcorr=\ones$. Hence, showing that $\priorv$ is a minimizer of a quadratic form $\zetav^{T}\cpmathat\zetav$ among all $\zetav\in\Delta(\types)$ is enough to prove the result. Noting that we can rely on the Lagrangian approach, at
$\zetav=\priorv$, the derivative of the quadratic form is collinear to $\ones$ and hence,
collinear to the space $\Delta\left(\types\right)$. Thus, first-order
conditions are satisfied. At the same time,
$\cpmathat$ is completely positive and thus positive semi-definite. Thus, second-order
conditions are satisfied. The result follows.
\end{proof}

\setcounter{example}{3}
\begin{example}[$\allocv^{ND}$ and  $\allocv^{FD}$  not on the boundary of $\repset$]\label{example:fd-nd-not-bd} Consider the case of binary types, $\types=\{\type_1,\type_2\}$. Denote by $\belief\in[0,1]$ the probability of type $\type_2$ and let $\prior=1/2$. Consider the following welfare function:
\begin{align*}
\welfare(\belief,\type_1)=\frac{\sin(2\pi\belief)}{2(1-\belief)},\  \welfare(\belief,\type_2)=\frac{\sin(4\pi\belief)}{2\belief},
\end{align*}
with  $\welfare(1,\type_1)$ and $\welfare(0,\type_2)$ defined by continuity as equal to $-\pi$ and $2\pi$, respectively. Given this welfare function, the truth-adjusted welfare function is
\begin{align*}\Reputation(\belief,\type_1)=\sin(2\pi\belief),\  \Reputation(\belief,\type_2)=\sin(4\pi\belief).\end{align*}
The corresponding indirect utility in the supporting Bayesian persuasion problem in the direction $\directionv=(\lambda_1,\lambda_2)$ is equal to
\begin{align*}\hat v_\direction(\belief)=\lambda_1 \sin(2\pi\belief)+\lambda_2 \sin(4\pi\belief).\end{align*}
For any $\directionv\in\reals^2\setminus\{0\}$, $\hat v_\direction(\prior)=\hat v_\direction(1/2)=\hat v_\direction(0)=\hat v_\direction(1)=0$. Hence, both full disclosure and no disclosure results in zero payoff. At the same time, for any such $\directionv$, $\hat v_\direction(\belief)$ is a non-constant continuous function anti-symmetric around $\mu=1/2$. Hence, it achieves strictly positive values on $[0,1]$ and $\cav \hat v_\direction(\belief_0)>0$ so that optimal disclosure outperforms both full disclosure and no disclosure. \autoref{theorem:hyperplanes}--extended to all boundary points--implies that $\allocv^{ND}$ and  $\allocv^{FD}$ are not on the boundary of \repset. 
\end{example}

The proofs of \autoref{proposition:bi-pooling} and \autoref{corollary:linear-pareto} rely on the graph-theoretic approach in \cite{rayo2010optimal}, the main properties of which we summarize in \autoref{remark:rs-appendix}:
\begin{remark}[\citealp{rayo2010optimal}]\label{remark:rs-appendix}
\cite{rayo2010optimal} propose the following graphical depiction of an \informationpolicy, \policy. Given a direction \directionv, let the prospect values $(\frac{\direction(\type_i)}{\prior(\type_i)},\rho(\type_i))=(\gammavi,\linearvi)$ for $i=1,\dots, N$ be vertices of a graph in $\reals^2$. Connect the points $(\gammavi,\linearvi)$ and $(\gammavj,\linearvj)$ by an edge if and only if a signal $s$ exists such that $\excdf(\signal\mid\type_j)\excdf(\signal\mid\type_i)>0$. The set of types that have positive probability under \signal\ is called the \emph{pooling} set of signal \signal.

    Lemmas 2--5 in \cite{rayo2010optimal} establish that under any optimal \informationpolicy, the following hold:
    \begin{enumerate}[label=(\alph*)]
        \item\label{itm:monotonicity} the posterior expectations of the prospect values induced by any two signals are ranked (in vector order), that is, for any two signals $s,s^\prime$, either $(\mathbb{E}_{\belief_s}[\gamma(\typec)],\mathbb{E}_{\belief_s}[\linear(\typec)])\geq(\mathbb{E}_{\belief_{s^\prime}}[\gamma(\typec)],\mathbb{E}_{\belief_{s^\prime}}[\linear(\typec)])$ or the opposite inequality holds.
        \item\label{itm:segment} prospects appear in the support of some signal only if they lie on a straight line with non-positive slope, 
        \item\label{itm:endpoints} if the pooling segments\footnote{By part \ref{itm:segment}, the pooling set of a signal lies on a \emph{segment}.} of two signals do not lie on the same line, they can intersect only if they share an endpoint, and
        \item\label{itm:order} if two prospect values are ranked  and appear in the support of two signals, then the posterior expectations induced by these signals are ranked in the same way. 
    \end{enumerate}
\end{remark}
\begin{proof}[Proof of Proposition \ref{proposition:bi-pooling}]
By the arguments presented in the main text, any optimal \informationpolicy\ solves the instance of the problem of \cite{rayo2010optimal} in which the prospect values are $(0,\linear(\type_j))$ for $j\neq i$ and $(1,\linear(\targettype))$ for the sender and for the receiver, respectively.

Given the structure of the prospect values in our problem, the property in part \ref{itm:segment} implies that $\type_i$ is never pooled with lower-index types. Furthermore, whenever it is pooled with some type, it is pairwise pooled. The property in part \ref{itm:endpoints} implies that whenever $\type_i$ is pooled with some type $\type_j$, then no types $\type_k$, $\type_l$ with $k<j<l$ can be pooled. Together with the property in part \ref{itm:order}, this observation implies that whenever $\type_i$ is pooled with some type $\type_j$, it is also pooled with all types $\type_k$ with $k>j$. Moreover, as $\type_i$ is pooled with increasingly higher-index types, the corresponding posterior expectations increase in vector order, which means that higher signals induce higher reputation yet have a relatively higher proportion of $\type_i$ (if all types are equally likely, then the probability of pooling $\type_i$ with $\type_j$ increases in $j$).

It is left to show that the threshold type--the lowest type with which $\type_i$ is pooled--is not pooled with any type of lower index. However, because the threshold type  is of higher index than $\type_i$, such pooling could clearly be improved by pooling the threshold type exclusively with $\type_i$.
\end{proof}

\begin{proof}[Calculations for Example \ref{example:linear}]
Define the following parameterized family of \informationpolicies\ (rows correspond to types and columns to signals):
\begin{gather*}
\Pi_{1}\left(\alpha,\beta\right)=\left(\begin{array}{ccc}
\alpha & 1-\alpha & 0\\
1-\beta & \beta & 0\\
0 & 0 & 1
\end{array}\right),\quad \Pi_{2}\left(\alpha\right)=\left(\begin{array}{cc}
\alpha & 1-\alpha\\
1 & 0\\
0 & 1
\end{array}\right),\\
\Pi_{3}\left(\beta\right)=\left(\begin{array}{cc}
1 & 0\\
0 & 1\\
\beta & 1-\beta
\end{array}\right),\quad \Pi_{4}\left(\alpha,\beta\right)=\left(\begin{array}{ccc}
1 & 0 & 0\\
0 & \alpha & 1-\alpha\\
0 & 1-\beta & \beta
\end{array}\right).
\end{gather*}
Note  \informationpolicies\ $\Pi_{1}\left(1,1\right)$ and $\Pi_{4}\left(1,1\right)$
coincide and correspond to full disclosure. Likewise, \informationpolicies\
$\Pi_{2}\left(0\right)$ and $\Pi_{3}\left(1\right)$ both correspond to full pooling of types $\type_{1}$ and $\type_{3}$. Information structures $\policy_2$ and $\policy_3$ are the $\type_1$-\bipooling\ and the $\type_3$-noisy-degrading policies, respectively. Like the $\type_3$-\bipooling, $\policy_1$ separates $\type_3$ from $\type_1$ and $\type_2$, but unlike the $\type_3$-\bipooling, it allows for $\type_1$ and $\type_2$ to be pooled, which does not affect $\type_3$'s expected reputation. Similarly, $\policy_4$ separates $\type_1$ from $\type_2$ and $\type_3$ like the $\type_1$-noisy-degrading policy, but unlike this policy, it allows for $\type_2$ and $\type_3$ to be pooled.
%


By \autoref{proposition:linear-hyperplanes}, any solution to the supporting Bayesian persuasion problem solves the instance of the problem of \cite{rayo2010optimal} with prospect  values  $\{(\gammavi,\linearvi):i\in\{1,2,3\}\}$  for the sender and for the receiver, respectively.  For simplicity, we assume that the types are strictly ranked under $\linear$, i.e., $\linearv_1<\linearv_2<\linearv_3$.

The property in part \ref{itm:segment} in \autoref{remark:rs-appendix} implies that if $(\gammavi,\linearvi)<(\gammavj,\linearvj)$, then types $\type_i$ and $\type_j$ are never pooled (cf. \autoref{corollary:linear-pareto}). We can then immediately establish the properties of the boundary \informationpolicies\ in the following cases:
\begin{itemize}[leftmargin=*,label=--]

\item If $\gammav_1<\gammav_2<\gammav_3$, then the uniquely optimal \informationpolicy\ is full disclosure.

\item If $\gammav_2<\gammav_1<\gammav_3$, then any optimal \informationpolicy\ separates type $\type_3$ and belongs to class $\Pi_{1}(\alpha,\beta)$.

\item If $\gammav_1<\gammav_3<\gammav_2$, then any optimal \informationpolicy\ separates type $\type_1$ and belongs to class $\Pi_{4}(\alpha,\beta)$.

\item If $\gammav_2<\gammav_3<\gammav_1$, then an optimal \informationpolicy\ never pools types $\type_2$ and $\type_3$ and belongs to class $\Pi_{2}(\alpha)$.

\item If $\gammav_3<\gammav_1<\gammav_2$, then an optimal \informationpolicy\ never pools types $\type_1$ and $\type_2$ and belongs to class $\Pi_{3}(\beta)$.
\end{itemize}
In the remaining case $\gammav_3<\gammav_2<\gammav_1$, no two prospects are ranked. However, by the property in part \ref{itm:monotonicity} in \autoref{remark:rs-appendix}, the induced posterior expectations  are necessarily ranked. Hence, if all three prospects lie on a straight line, then no disclosure is optimal. In contrast, if the three prospects do not lie on a straight line, then an optimal \informationpolicy\  separates either types $\type_1$ and $\type_2$ or types $\type_2$ and $\type_3$, and thus belongs to either class $\Pi_{2}(\alpha)$ or to class $\Pi_{3}(\beta)$.

Finally, it is easy to see that by the same arguments, an optimal \informationpolicy\ for the cases in which $\gammavi=\gammavj$ for some $i$ and $j$ belongs to one of the same four classes of \informationpolicies. 

Knowing the classes of boundary \informationpolicies, 
 we can plot the \bset\ in \autoref{example:linear} by direct calculation.
\end{proof}
 \newpage
 \begin{center}\Large\textsc{Online Appendix}\end{center}
 \section{Data Limits}\label{sec:cohort}
 The analysis in the paper assumes that the \informationpolicy\ can arbitrarily condition on an \agent's payoff-relevant \characteristic. However, this assumption does not necessarily hold in many applications of interest. For instance, regulation may prevent the disclosure of protected characteristics, such as gender or race. Thus, when \type\ encompasses such characteristics, considering information structures that respect these restrictions is natural.

     In this section, we extend our analysis by removing this assumption. Formally, we consider the following extension of the model in \autoref{sec:model}. Together with the \agents' \characteristics, we are given a data source that is potentially informative about these \characteristics. The data source has realizations in a finite set $\Data\equiv\{\data_1,\dots,\data_M\}$. We describe the joint distribution over payoff-relevant \characteristics\ and data via the prior distribution on \Types, \prior, and a system of conditional probabilities $\{\jointprior(\cdot|\type):\type\in\Types\}$, describing the distribution of the data source \data\ conditional on the \characteristics\ \type. We let $\priordata\in\Delta(\Data)$ denote the induced marginal distribution on \Data.\footnote{Although we should index \priordata\ by $\left(\prior,\left(\jointprior(\cdot|\type)\right)_{\type\in\types}\right)$, we omit this dependence to simplify notation.} The model in \autoref{sec:model} corresponds to the case in which $\Types=\Data$ and $\jointprior(\data|\type)=\mathbbm{1}[\data=\type]$.
    
  We assume information can be provided to the outside observer only about data-source realizations, but not an \agent's type. Formally, an \informationpolicy\ $\policy=(\excdf,\signals)$ consists of a countable set of labels \signals\ and a mapping \excdf, which associates to each data-source realization, \data, a distribution over signals $\excdf(\cdot|\data)\in\Delta(\signals)$. Given an \informationpolicy\ $\policy$ and a signal realization $\signal\in\signals$,  updated beliefs about \type\ depend only on the updated belief about the realization of \data. Indeed,
 \begin{align}\label{eq:belief-data-type}
 \belief_s(\type)=\sum_{\data\in\Data}\frac{\prior(\type)\jointprior(\data|\type)}{\priordata(\data)}\beliefdata_s(\data),
 \end{align}
 where $\posteriordata_s$ is the marginal on \Data\ of the updated joint belief on $\types\times\Data$. It follows that we can define the welfare function as depending on beliefs about \data\ rather than about \type. That is, we can define the function $\weta:\Delta(\Data)\times\Types\mapsto\reals$ as follows:
 \begin{align*}
     \weta(\posteriordata,\type)=\welfare(\belief(\posteriordata),\type),
 \end{align*}
 where the function $\belief(\posteriordata)$ is determined by \autoref{eq:belief-data-type}.

 Given an \informationpolicy\ $(\excdf,\signals)$, the welfare of an individual of type \type\ is
 \begin{align}\label{eq:ip-payoff-general}
 \payoff_{\policy}(\type)\equiv &\mathbb{E}_{\conditional}[\weta(\tilde{\posteriordata},\type)]=\sum_{s\in\signals}\sum_{\data\in\Data}\jointprior(\data|\type)\excdf(s|\data)\weta(\posteriordata_s,\type),
 \end{align}
 and the Bayes welfare set continues to be defined as the set of \bprofiles.

 We now show  the analysis in the main text extends verbatim. Indeed, by the same arguments as in \autoref{sec:main}, the welfare of an individual of type \type\ under \informationpolicy\ $\policy=(\excdf,\signals)$ can be written as:
 \begin{align}\label{eq:accounting-general}
 \payoff_{\policy}(\type)&=\mathbb{E}_{\conditional}[\weta(\tilde{\posteriordata},\type)]=\expect_{\unconditional}[\wetahat(\tilde{\posteriordata},\type)],
 \end{align}
 where the truth-adjusted welfare function $\wetahat$ now takes the form:
 \begin{align}\label{eq:adjusted-payoff-general}
 \wetahat(\posteriordata,\type)=\sum_{\data\in\Data}\jointprior(\data|\type)\frac{\posteriordata(\data)}{\priordata(\data)}\weta(\posteriordata,\type).
 \end{align}
 By separating the variable on which welfare is conditioned on--the payoff-relevant \characteristics, \type--from the variable about which information is provided--the data source, \data--\autoref{eq:adjusted-payoff-general} allows us to provide further insight into the truth-adjusted welfare function in the model in \autoref{sec:model}. Indeed, note the likelihood correction is based on the variable \data, highlighting that it corresponds to the variable about which information is provided. Similar to before, we can interpret the likelihood-ratio adjustment as describing that each data-source realization \data\ has a budget $\priordata(\data)$ to be distributed across different (data) posteriors \posteriordata. Unlike the analysis before, individuals of type \type\ only own a fraction $\jointprior(\data|\type)$ of this ratio. 

 \autoref{eq:accounting-general} implies \autoref{theorem:concavification} immediately extends  to this setting:
 \begin{theorem}\label{theorem:cohort}
 The \bset\ \repset\ satisfies the following:
 \begin{align}
     \repset=\left\{\payoffv\in\mathbb{R}^N:(\priordata,\payoffv)\in\mathrm{co}\left(\mathrm{graph}\,\wetahatv\right)\right\}.
 \end{align}
 \end{theorem}
 In what follows, we explore how the \bset\ changes as we change the informativeness of the data source. Intuitively, we would expect that the \bset\ shrinks as data becomes \emph{less precise}. \autoref{proposition:Blackwell} below shows that this intuition holds when the notion of less precise coincides with the notion of \emph{garbling} in \cite{blac51}.

 Formally, given the distribution of payoff-relevant \characteristics\ $\prior\in\Posteriors$,  we wish to understand the effect of different data sources, as described by data-source realizations \Datab\ and conditional probability systems $\{\jointpriorb(\cdot|\type)\in\Delta(\Datab):\type\in\Types\}$. Following \cite{blac51}, we say $(\Datab,\jointpriorb)$ is a \emph{garbling} of $(\Data,\jointprior)$ if a stochastic matrix $G:\Data\mapsto\Delta(\Data^\prime)$ exists such that for every data-type pair $(\data^\prime,\type)$, 
 \begin{align*}
 \jointpriorb(\data^\prime|\type)=\sum_{\data\in\Data}G(\datab|\data)\jointprior(\data|\type).
 \end{align*}
 Let $\repset(\prior,\welfare,\Data,\jointprior)$ denote the \bset\ for prior type distribution \type\ and welfare function \welfare, as we vary the (informativeness of the) data source $(\Data,\jointprior)$. We then have the following:
 \begin{proposition}[Data Comparison]\label{proposition:Blackwell}
 $\repset(\prior,\welfare,\Datab,\jointpriorb)\subseteq\repset(\prior,\welfare,\Data,\jointprior)$ for all welfare functions \payoff\ and type distributions \prior\ if and only if $(\Datab,\jointpriorb)$ is a garbling of $(\Data,\jointprior)$.
 \end{proposition}

 \begin{proof}[Proof of \autoref{proposition:Blackwell}]
 One direction is straightforward: if $(\Datab,\jointpriorb)$ is a garbling of $(\Data,\jointprior)$,  any distribution of signals conditional on payoff-relevant \characteristics\ induced by some information structure under data  source $(\Datab,\jointpriorb)$ is feasible under data source $(\Data,\jointprior)$. Consequently, any welfare profile that can be induced by some information structure under $(\Datab,\jointpriorb)$ can be induced under  $(\Data,\jointprior)$.

 To obtain the other direction, 
  toward a contradiction, assume $(\Datab,\jointpriorb)$ is not a garbling of $(\Data,\jointprior)$. Then, by \cite{blac51}, a prior $\prior\in\Delta(\types)$ and a payoff function $u:A\times\Types\rightarrow\reals$ exist such that a decision maker  with utility $u$ derives strictly greater value from having access to $(\Datab,\jointpriorb)$ than to $(\Data,\jointprior)$. That is, letting $U(\belief)$ denote the decision maker's indirect utility, $ \max_{a\in A}\expect_{\belief}[u(a,\type)]$, we have that:
 \begin{align}\label{eq:Blackwell-proof-1}
     \sum_{\type\in\types}\prior(\type) \expect_{(\Datab,\jointpriorb)}[U(\belief)\mid \type]>\sum_{\type\in\types}\prior(\type) \expect_{(\Data,\jointprior)}[U(\belief)\mid \type].
 \end{align}
 Consider now the \bsets\ given  the welfare function $\payoff(\belief,\type)=U(\belief)$ and prior distribution $\prior$, under data sources $(\Data,\jointprior)$ and $(\Datab,\jointpriorb)$. We have that 
 \begin{align}
 \sum_{\type\in\types}\prior(\type) \expect_{(\Datab,\jointpriorb)}[U(\belief)\mid \type]&=\max_{\allocv\in \repset(\cdot,\Datab,\jointpriorb)}\sum_{\type\in\types}\prior(\type)\allocv(\type)\label{eq:Blackwell-proof-2} \\
 &\leq \max_{\allocv\in \repset(\cdot,\Data,\jointprior)}\sum_{\type\in\types}\prior(\type)\allocv(\type)=\sum_{\type\in\types}\prior(\type) \expect_{(\Data,\jointprior)}[U(\belief)\mid \type],\nonumber
 \end{align}
 where the equalities follow because the maximal ex ante payoff is obtained by having full access to available data, and the inequality follows from the assumption that $\repset(\prior,\welfare,\Datab,\jointpriorb)\subseteq \repset(\prior,\welfare,\Data,\jointprior)$ for all $\payoff$ and $\prior$. Comparing Equations \ref{eq:Blackwell-proof-1} and \ref{eq:Blackwell-proof-2} leads to the desired contradiction and the result follows. 
 \end{proof}

 \begin{remark}[When to blind an algorithm]\label{remark:liang}
 \autoref{proposition:Blackwell} stands in contrast with the recommendation in the algorithmic fairness literature to ``blind'' algorithms to sensitive inputs such as race or gender. Indeed, having taken into account the impact of the outside observer's incentives in the population's welfare, allowing the \informationpolicy\ to condition on the \agents' payoff-relevant \characteristics\ leads to the largest \bset, thereby (weakly) increasing the value of any social welfare function that is used to choose what \informationpolicy\ to implement.

 There may be other reasons outside our model that could justify blinding the information structure to the \agents' types. For instance, \cite{liang2022algorithmic} show that when an agent different from the social planner selects a decision-making algorithm, the social planner may prefer to restrict the inputs into the agent's algorithm. Even though in our model algorithms are information structures that send non-binding recommendations, a similar result would hold in our setting.
 \end{remark}

 We conclude this section with an example that illustrates the following two points. First, whereas \autoref{proposition:Blackwell} shows that less precise data sources limit the ability to generate and distribute welfare via information, the example shows that this effect is not uniform across \agents\ of different types. Second, the \bset\ may collapse to the no-disclosure \bprofile\ for data sources that are strictly more informative than no information in the Blackwell order. 

 \setcounter{example}{1}
 \begin{example}[\autoref{example:platforms} continued; Noisy Data]\label{example:Blackwell-platform}
 Suppose the online marketplace only has access to a noisy estimate of the consumer's type, perhaps from past purchases or undeleted cookies. We model this as a data source that reveals a consumer's type with a fixed precision $\precision\in[1/2,1]$: $\Data=\{\data_1,\data_2\}$ and $\jointprior(\data_i|\type_i)=\precision$. When $\precision=1$, the data source is perfectly informative about a consumer's type; when $\precision=1/2$, the data source is pure noise. More generally, if $\precision<\precision^\prime$,  the data source that corresponds to $\precision$ is a garbling of the data source  that corresponds to $\precision^\prime$.
 \begin{figure}[t!]
 \centering\scalebox{0.8}{%
 \begin{tikzpicture}[scale=1.2]
 \begin{axis}[xtick pos=left,ytick pos=left,
 width=8cm, height=8cm, 
             xmin=-0.1,xmax=0.9,
         ymin=-0.1,ymax=0.75,
         xtick={0,0.5},
         ytick={0,0.5},
         xlabel=$\allocv_H$,ylabel=$\allocv_L$
          ,
          x label style={at={(0.8,0.01)}},
     y label style={at={(0.09,0.8)},rotate=-90}
 ]
 \addplot +[name path=A,mark=none,color=blue, solid, thick] coordinates {(0,0.6) (0.7, 0.6)};
 \addplot +[mark=none,color=blue, solid, thick] coordinates {(0.8,0.4) (0.7, 0.6)};
 \addplot +[name path=C,mark=none,color=blue, solid,thick] coordinates {(0.6,0.2) (0.8, 0.4)};
 \addplot +[solid,name path=D, mark=none,color=blue,dashed, thick] coordinates {(0.152,0.101) (0.6,0.2)};
 \addplot +[solid,mark=none,color=blue, thick] coordinates {(0,0.6) (0, 0.202)};
 \addplot +[name path=B,dashed,mark=none,color=blue, thick] coordinates {(0.152,0.101) (0, 0.202)};
 \addplot [blue!30] fill between [of = A and B, soft clip={domain=0:1}];
 \addplot [blue!30] fill between [of = A and C, soft clip={domain=0:1}];
 \addplot [blue!30] fill between [of = A and D, soft clip={domain=0:1}];


 \addplot +[name path=A,mark=none,color=blue, solid,thick] coordinates {(0.210, 0.510) (0.630, 0.542)};
 \addplot +[mark=none,color=blue, thick,solid] coordinates {(0.630, 0.542) (0.8,0.4)};
 \addplot +[name path=C,mark=none,color=blue, solid,thick] coordinates {(0.6,0.2) (0.8,0.4)};
 \addplot +[solid,name path=D, mark=none,color=blue,dashed, thick] coordinates {(0.152,0.101) (0.6,0.2)};
 \addplot +[solid,mark=none,color=blue, thick] coordinates {(0.210, 0.510) (0.0938,0.228)};
 \addplot +[name path=B,dashed,mark=none,color=blue, thick] coordinates {(0.152,0.101) (0.0938,0.228)};
 \addplot [blue!45] fill between [of = A and B, soft clip={domain=0:1}];
 \addplot [blue!45] fill between [of = A and C, soft clip={domain=0:1}];
 \addplot [blue!45] fill between [of = A and D, soft clip={domain=0:1}];


 \addplot +[name path=A,mark=none,color=blue, solid,thick] coordinates {(0.420, 0.420) (0.508,0.436)};
 \addplot +[mark=none,color=blue, thick,solid] coordinates {(0.508,0.436) (0.8,0.4)};
 \addplot +[name path=C,mark=none,color=blue, solid,thick] coordinates {(0.6,0.2) (0.8,0.4)};
 \addplot +[solid,name path=D, mark=none,color=blue,dashed, thick] coordinates {(0.151,0.100) (0.6,0.2)};
 \addplot +[name path=B,solid,mark=none,color=blue, thick] coordinates {(0.420, 0.420) (0.281, 0.281)};
 \addplot +[dashed,mark=none,color=blue, thick] coordinates {(0.151,0.100) (0.281, 0.281)};

 \addplot [blue!60] fill between [of = A and C, soft clip={domain=0:1}];
 \addplot [blue!60] fill between [of = B and D, soft clip={domain=0:1}];

 \node at (axis cs:0.1,0.55) {\large$\precision=1$};
 \node at (axis cs:0.32,0.45) {\large$\precision=0.85$};
 \node at (axis cs:0.52,0.3) {\large$\precision=0.7$};
 \addplot[color=blue,mark=*]coordinates{(0.6,0.2)}
 [every node/.style={xshift=0pt}]
 node[right] {\large$\color{black}{\precision=0.55}$};

 \end{axis}
 \end{tikzpicture}  }
         \caption{Noisy data in the online marketplace. \bsets\ for different values of $\precision\in\{0.55, 0.7, 0.85, 1\}$.}\label{fig:Blackwell}
 \end{figure}

 \autoref{fig:Blackwell} illustrates the \bset\ \repset\ for different precision values. Three features are worth noting. First, in line with \autoref{proposition:Blackwell}, \bsets\ resulting from data sources with lower precision are subsets of those with higher precision. When $\precision=1$, the \bset\ naturally coincides with the one in \autoref{fig:admissions} in the main text. Second, at high values of $\precision$, lower data precision has asymmetric effects across types: it decreases the maximal payoff of \typel-consumers without affecting their minimal payoff, yet it increases the minimal payoff of \typeh-consumers without affecting their maximal payoff. Indeed, for sufficiently low values of \precision, the unique Pareto efficient information structure is the one that maximizes the payoff  of \typeh-consumers. That is, in this example, 
 lower data precision benefits \typeh-consumers. Finally, whereas it is immediate that the \bset\ coincides with the no-disclosure profile $\allocv^{ND}$ when $\precision=1/2$, the \bset\ actually collapses to this point at $\precision=3/5$: Once $\precision<3/5$,  generating Bayes plausible distributions over posteriors with support outside the interval $[1/2,3/4)$ is not possible, and on this interval, \payoff\ is constant. This feature highlights that an incrementally more informative data source may have a discontinuous impact on welfare redistribution possibilities. \hfill$\closing$
 \end{example}
\end{document}